\documentclass[12pt]{article}

\usepackage{graphicx}
\usepackage{color}

\RequirePackage{amsthm, amsmath, amsfonts, amssymb}%

\usepackage{amsmath}
\usepackage{amssymb}
\usepackage{amscd}

\usepackage[latin1]{inputenc}

\usepackage[boxruled]{algorithm2e}

\newcommand{\R}{\mathbb{R}}
\renewcommand{\P}{\mathbb{P}}
\newcommand{\inr}[1]{\bigl< #1 \bigr>}
\newcommand{\Inr}[1]{\left\langle#1\right\rangle}

\newcommand{\sgn}{{\rm sgn}}

\newcommand{\E}{\mathbb{E}}

\newcommand{\eps}{\varepsilon}
\newcommand{\conv}{\mathop{\rm conv}}

\newcommand{\po}{\cP_\Omega}
\newcommand{\grad}{\nabla}
\DeclareMathOperator{\diag}{diag}
\DeclareMathOperator{\supp}{supp}

\DeclareMathOperator*{\tr}{tr}
\DeclareMathOperator*{\ra}{rank}

\DeclareMathOperator{\im}{im}

\DeclareMathOperator*{\argmin}{argmin}

\DeclareMathOperator*{\rank}{rank}

\DeclareMathOperator{\prox}{prox}

\newcommand{\Pro}{\mathbb P}

\newcommand{\1}{{\rm 1}\kern-0.24em{\rm I}}

\newcommand{\cP}{{\mathcal P}}

\newcommand \cS{{\cal S}}

\newcommand{\norm}[1]{\|#1\|}%
\newcommand{\Norm}[1]{\left|#1\right|}%

\newtheorem{Theorem}{Theorem}[section]
\newtheorem{Lemma}[Theorem]{Lemma}

\newtheorem{theorem}{Theorem}%
\newtheorem{proposition}{Proposition}%
\theoremstyle{assumption}%
{\bf}{\rm}%
\theoremstyle{remark}%
\newtheorem{remark}{Remark}%

\numberwithin{equation}{section} 

\begin{document}

\title{Weighted algorithms for compressed sensing and matrix
  completion}

\author{St{\'e}phane Ga{\"i}ffas${}^{1,3}$ \and Guillaume
  Lecu\'e${}^{2, 3}$}

\footnotetext[1]{Universit\'e Pierre et Marie Curie - Paris~6,
  Laboratoire de Statistique Th\'eorique et Appliqu\'ee. \emph{email}:
  \texttt{stephane.gaiffas@upmc.fr}}

\footnotetext[2] {CNRS, Laboratoire d'Analyse et Math\'ematiques
  appliqu\'ees, Universit\'e Paris-Est - Marne-la-vall\'ee
  \emph{email}: \texttt{guillaume.lecue@univ-mlv.fr}}

\footnotetext[3]{This work is supported by French Agence Nationale de
  la Recherce (ANR) ANR Grant \textsc{``Prognostic''}
  ANR-09-JCJC-0101-01.
}

\maketitle

\begin{abstract}
  This paper is about iteratively reweighted basis-pursuit algorithms
  for compressed sensing and matrix completion problems. In a first
  part, we give a theoretical explanation of the fact that reweighted
  basis pursuit can improve a lot upon basis pursuit for exact
  recovery in compressed sensing. We exhibit a condition that links
  the accuracy of the weights to the RIP and incoherency constants,
  which ensures exact recovery. In a second part, we introduce a new
  algorithm for matrix completion, based on the idea of iterative
  reweighting. Since a weighted nuclear ``norm'' is typically
  non-convex, it cannot be used easily as an objective function. So,
  we define a new estimator based on a fixed-point equation. We give
  empirical evidences of the fact that this new algorithm leads to
  strong improvements over nuclear
  norm minimization on simulated and real matrix completion problems. \\

  \noindent%
\end{abstract}

\section{Introduction}
\label{sec:introduction}

In this paper, we consider the statistical analysis of high
dimensional structured data in two close setups: vectors with small
support and matrices with low rank. In the first setup, known as
Compressed Sensing (CS)
\cite{MR2241189,MR1639094,MR2243152,MR2236170,IEEE-Dononho,IEEE-CT},
the aim is to reconstruct a high dimensional vector with only few
non-zero coefficients, based on a small number of linear
measurements. In the second setup, called Matrix Completion
\cite{MR2723472,fazel2002matrix,candes-recht08,Gross}, we aim at
reconstructing a small rank matrix from the observations of only a few
entries. Both problems are motivated by many practical applications in
many different domains (medical~\cite{medical},
imaging~\cite{MR1440119}, seismology~\cite{Sismo}, recommending
systems such as the Netflix Prize, etc.) as well as theoretical
challenges in many different fields of mathematics (random matrices,
geometry of Banach spaces, harmonic analysis, empirical processes
theory, etc.).  From an algorithmic viewpoint, one central idea is the
convex relaxation of the $\ell_0$-functional (the function giving the
number of non-zero coefficients of a vector) and of the rank
function. This idea gave birth to two well-known algorithms: the Basis
Pursuit algorithm~\cite{MR1639094}~ and nuclear norm
minimization~\cite{candes-recht08}. Many results have been obtained
for these two algorithms and we refer the reader to the next sections
for more details. Here we will be interested in weighted versions of
these algorithms, see~\cite{MR2461611} in the CS setup. In particular,
we will be interested in finding theoretical explanation underlying
the fact that, empirically, it is observed that weighted Basis pursuit
outperforms classical Basis Pursuit. We will also propose a way to
export the idea of reweighting into the Matrix Completion problem.

\section{Weighted basis-pursuit in Compressed Sensing}
\label{sec:cs-vectors}

One way of setting the CS problem is to ask the following question.
Starting with a $m \times N$ matrix $A$, called a \emph{sensing} or
\emph{measurement} matrix, and with a vector $x$ in $\R^N$, is it
possible to reconstruct $x$ from the linear measurements $Ax$?
Classical linear algebra theory tells that we need at least $m \geq N$
to recover $x$ from $Ax$ in order to find a unique solution to the
linear system. But, if more is known on $x$, then, hopefully, a
smaller number $m$ of measurements may be enough.

In the theory of CS, it is now well-understood that it is indeed
possible to recover sparse signals (signals with a small support, the
support being the set of non-zeros entries) from a small number of
linear measurements. If $x$ is a sparse vector and $A$ a ``good''
measurement matrix (in a sense to be clarified later), then looking
for a vector $y$ with the smallest support and satisfying $Ay = Ax$
can recover $x$ exactly. This procedure, called the $\ell_0$ or
support minimization procedure, is known to be the best theoretical
procedure to recover any $s$-sparse vector $x$ (vectors with a support
size smaller than $s$) from $Ax$ as long as $A$ is injective on the
set of all $s$-sparse vectors. However, this problem is NP-hard, and
alternatives are suitable in practice, in part because the function $x
\mapsto |x|_0$ ($|x|_0$ stands for the cardinality of the support of
$x$) is not convex.

A natural remedy to this problem is convex
relaxation. In~\cite{MR1639094}, the authors propose to minimize the
$\ell_1$-norm as the convex envelope of this non-convex function,
leading to the so-called Basis-Pursuit algorithm (BP). The BP
algorithm minimizes the $\ell_1$ norm on the affine space $x + \ker
A$. Namely, consider, for any $y \in \R^m$:
\begin{equation}
  \label{eq:BP}
  \Delta_1(y) \in \argmin_{t \in \R^N} \Big(|t|_1 : At=y\Big),
\end{equation}
so that $\Delta_1(Ax)$ is a candidate for the reconstruction of $x$
based on $A x$. We say that $x$ is exactly reconstructed by
$\Delta_1$, namely $\Delta_1(Ax) = x$, when $x$ is the unique solution
of the minimization problem~\eqref{eq:BP} when $y=Ax$.

Note that other algorithms have been introduced in the CS
literature. For instance, $\ell_p$-minimization algorithms for $0<p<1$
are considered
in~\cite{MR2421974,MR2503311,MR2647010,chartrand2008iteratively}. Some
greedy algorithms based on the ideas of the Matching Pursuit algorithm
of~\cite{MR1321432,MatchingPursuit} have been used in CS,
see~\cite{MR2502366,MR2496554,MR2446929} for instance.

In the present paper, we consider weighted-$\ell_1$ minimization over
$x+\ker A$. This algorithm was introduced in~\cite{MR2461611}. Since
then, it has drawn a particular attention because it is now
acknowledged, although mainly only empirically observed, that a proper
weighted basis-pursuit algorithm can improve a lot upon basic
basis-pursuit. This is illustrated in Figure~\ref{fig:phase-cs}, and
many other numerical experiments can be found
in~\cite{MR2461611}. However, theoretical explanations of this fact
are still lacking. Some results that go in this direction are given
in~\cite{2009arXiv0901.2912A, 2009arXiv0904.0994X,
  khajehnejadanalyzing}, \cite{chartrand2008iteratively},
\cite{2009arXiv0901.2912A}. But, the results given in these papers are
of a different nature than ours, since they are using a random model
for the unknown vector $x$, such as a vector with i.i.d $N(0, 1)$
non-zero entries, with a distribution support which is uniform
conditionally on the sparsity. In the statement of our results, $x$ is
an arbitrary deterministic sparse vector. In~\cite{MR2588385} an
iteratively reweighted least-squares procedure is studied, as an
approximation of basis-pursuit.

We introduce the weighted algorithm: for any $y \in \R^m$ and any
sequence $w = (w_1, \ldots, w_N)\in\R^N$ of non-negative weights,
\begin{equation}
  \label{eq:general-weighted-algo}
  \Delta_w(y) \in \argmin_{t\in\R^N} \Big(\sum_{i=1}^N
  \frac{|t_i|}{w_i} : At = y \Big).
\end{equation}
We use the convention $t / 0 = \infty$ when $t > 0$ and $0 / 0 = 0$.
Note that, under this convention, the
algorithm~\eqref{eq:general-weighted-algo} is defined according to the
support $I_w$ of $w$ by
\begin{equation}
  \label{eq:values-weight-algo}
  \big( \Delta_w(y) \big)_{I_w^c} = 0 \; \text{ and } \;
  \big( \Delta_w(y) \big)_{I_w} \in \argmin_{t \in \R^{I_w}}
  \Big(\sum_{i\in I_w} \frac{|t_i|}{w_i} : A_{I_w} t = y \Big),
\end{equation}
where if $t \in \R^N$ and $I \subset \{ 1, \ldots, N \}$, we denote by
$t_I$ the vector such that $(t_I)_i = t_i$ if $i \in I$ and $(t_I)_i =
0$ if $i \notin I$. Once again, we say that $x$ is exactly
reconstructed by $\Delta_w$, namely $\Delta_w(Ax) = x$, when $x$ is
the unique solution of the minimization
problem~\eqref{eq:general-weighted-algo} when $y=Ax$. In particular,
this requires that the support of $x$ is included in the support of
$w$.

\subsection{No-loss property}

Note that when the weight vector $w$ is close to $x$, then
$\sum_{i=1}^N |x_i| / w_i$ is close to $|x|_0$. Moreover, for
``reasonable'' matrices $A$, the vector $x$ is the one with the
shortest support in the affine space $x + \ker A$. So, a natural
choice for $w$ in~\eqref{eq:general-weighted-algo} is $w =
|\Delta_1(Ax)|$. We denote this decoder by $\Delta_2$:
\begin{equation}
  \label{eq:weighted-BP-2}
  \Delta_2(y) \in \argmin_{t \in \R^N}
  \Big(\sum_{i=1}^N \frac{|t_i|}{|\Delta_1(y)_i|} : At = y \Big).
\end{equation}
The next Theorem proves that $\Delta_2$ is at least as good as the
Basis Pursuit algorithm $\Delta_1$.
\begin{theorem}
  \label{thm:A}
  Let $x \in \R^N$. If $\Delta_1(Ax)=x$, then $\Delta_2(Ax)=x$.
\end{theorem}
The proof of Theorem~\ref{thm:A} is based on the well-known null space
property and dual characterization of~\cite{MR2236170}, see
Section~\ref{sec:proofs} below. However, it was observed empirically
in~\cite{MR2461611} that it is better to consider positive weights,
and thus, to consider, for some $\eps > 0$, the weights $w_i =
|\Delta_1(y)_i| + \eps$ for $i= 1, \ldots, N$. This is easily
understood: if for some $i\in\{1,\ldots,N\}$, $\Delta_1(Ax)_i=0$ while
$x_i\neq 0$, then $\Delta_2(Ax)_i$ is also equal to $0$ and there is
no hope to recover $x$ using $\Delta_2$ as well. By adding an extra
$\eps$ term to each weights, the necessary support condition $\supp(x)
\subset \supp(w)$ to reconstruct $x$ from $\Delta_w(Ax)$ is satisfied
(see for instance
Proposition~\ref{prop:equivalence-reconstruction-exacte} in
Section~\ref{sec:proofs}). The choice of $\eps > 0$ can be done in a
data-driven way, see~\cite{MR2461611}.

\subsection{An empirical evidence}
\label{sec:empirical-evidence}

In Figure~\ref{fig:phase-cs}, we give a simple illustration of the
fact that weighted basis-pursuit can improve a lot upon basic
basis-pursuit, using a simple numerical experiment. For many
combinations of $m$ ($y$-axis) and $s$ ($x$-axis), we repeat the
following experiment 50 times: draw at random a sensing matrix $A$
with i.i.d $N(0, 1/m)$ entries and draw at random a vector with $s$
non-zero coordinates chosen uniformly, with i.i.d $N(0, 1)$ non-zero
entries. Then, compute $\hat x_1 = \Delta_1(Ax)$ and $\hat x_w =
\Delta_{20}^\eps(Ax)$ (here we take $\eps = 0.01$ without further
investigation), where $\Delta_k^\eps(Ax)$ is computed iteratively,
using
\begin{equation}
  \label{eq:delta-k-def}
  \Delta_{k+1}^\eps(Ax) \in \argmin_{t \in \R^N}
  \Big(\sum_{i=1}^N \frac{|t_i|}{|\Delta_k^\eps(Ax)_i| + \eps} : At = Ax \Big).
\end{equation}
Then, we count the number of exact reconstructions achieved by $\hat x_1$
and $\hat x_w$ over the 50 repetitions. The plots on the left are the
exact recovery counts of $\hat x_1$ (black means exact recovery over
the 50 repetitions) while the plots on the right are the exact
recovery counts of $\hat x_w$. In these figures, exact recovery is
declared exact when $|\hat x - x|_2 / |x|_2 < \eta$, where we take
$\eta = 10^{-5}$ on the first line and $\eta = 10^{-6}$ on the second
line. The red curve is a theoretical ``phase-transition'' threshold $s
\mapsto s \log(e m / s)$. We observe in these figures that $\hat x_w$
improves a lot upon $\hat x_1$, in particular when $\eta = 10^{-6}$.

\newlength{\figwidth}
\newlength{\figlength}
\setlength{\figwidth}{7cm}

\begin{figure}[htbp]
  \centering
  \includegraphics[width=\figwidth]{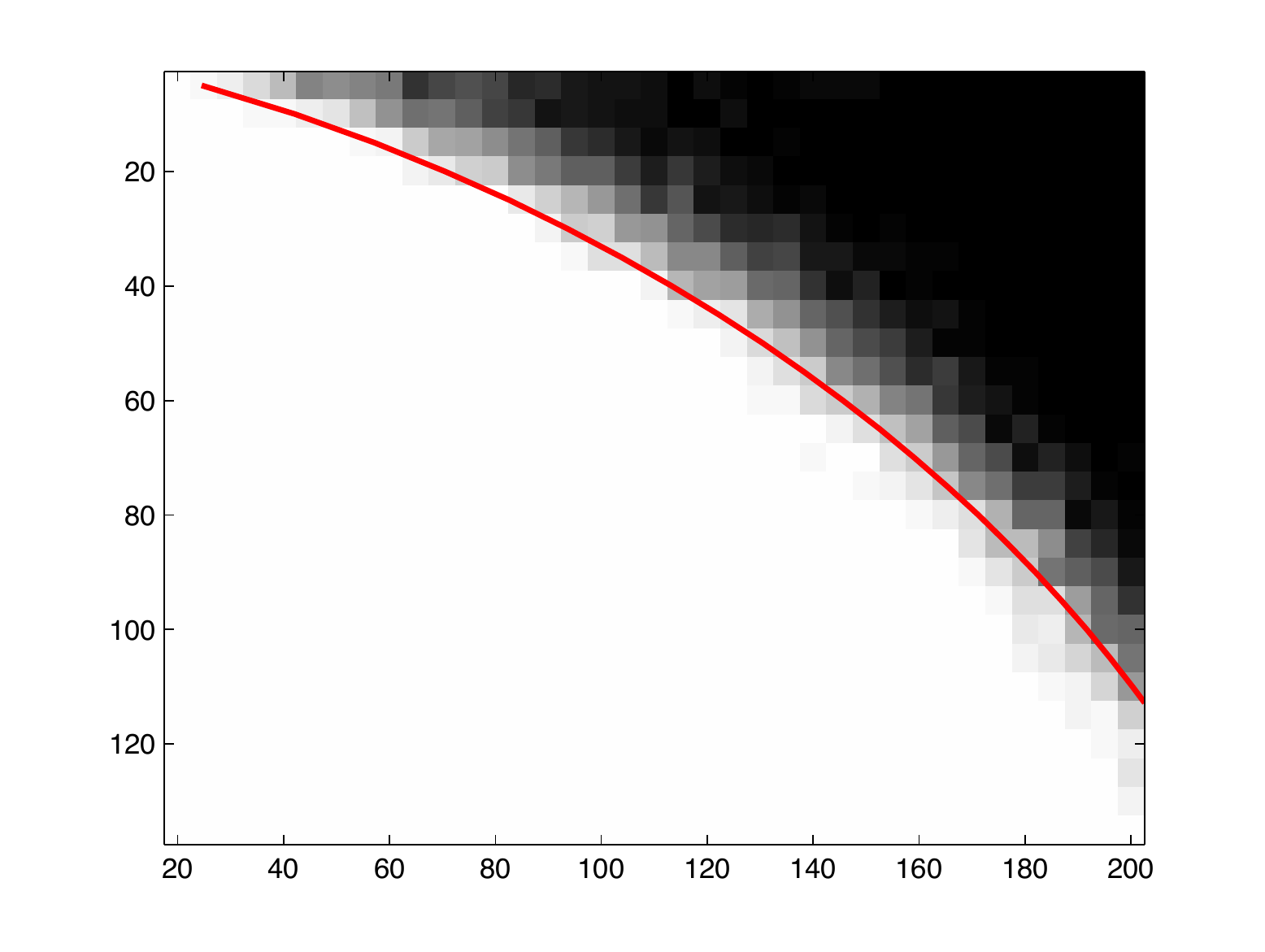}%
  \includegraphics[width=\figwidth]{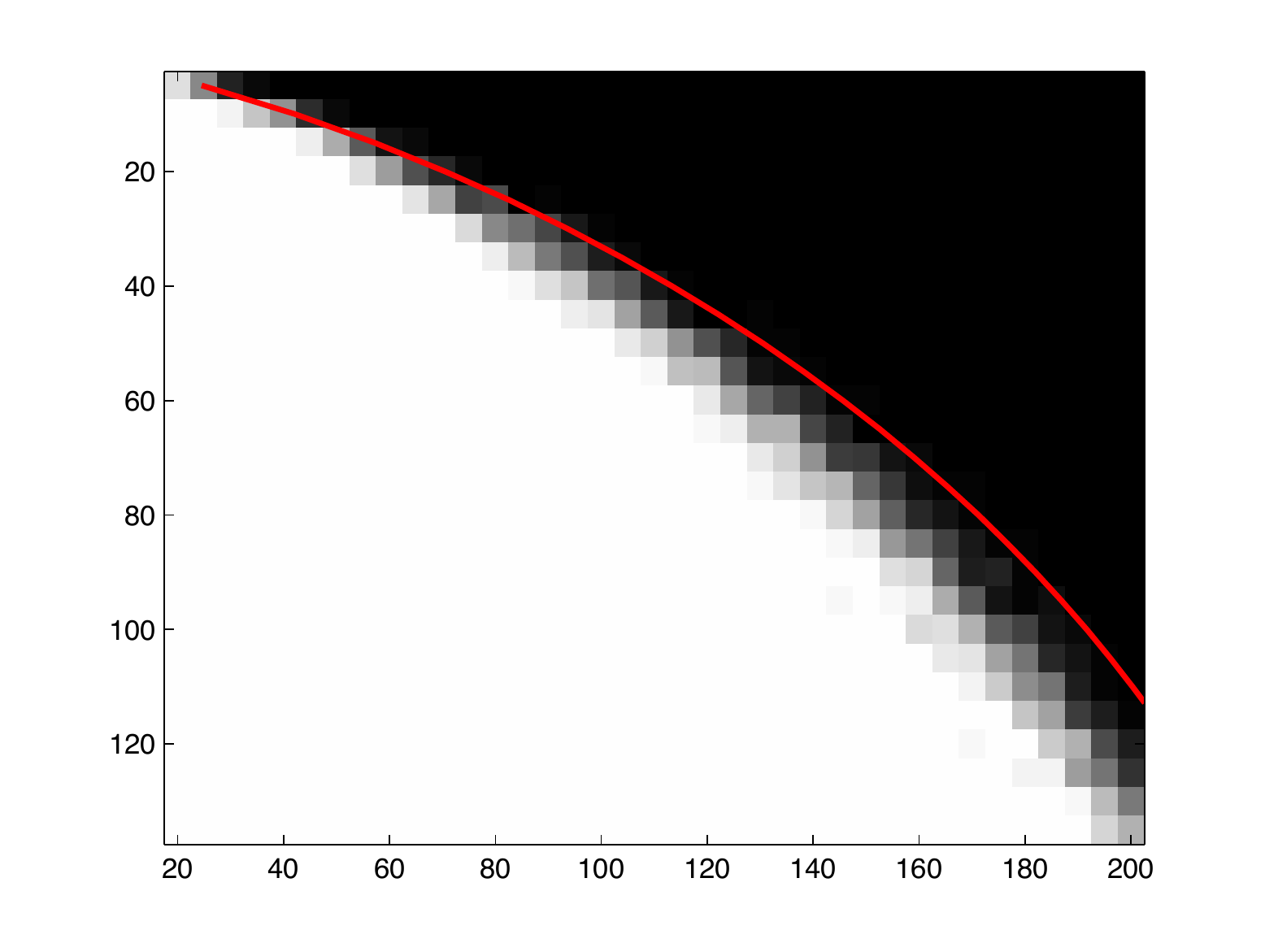} \\
  \includegraphics[width=\figwidth]{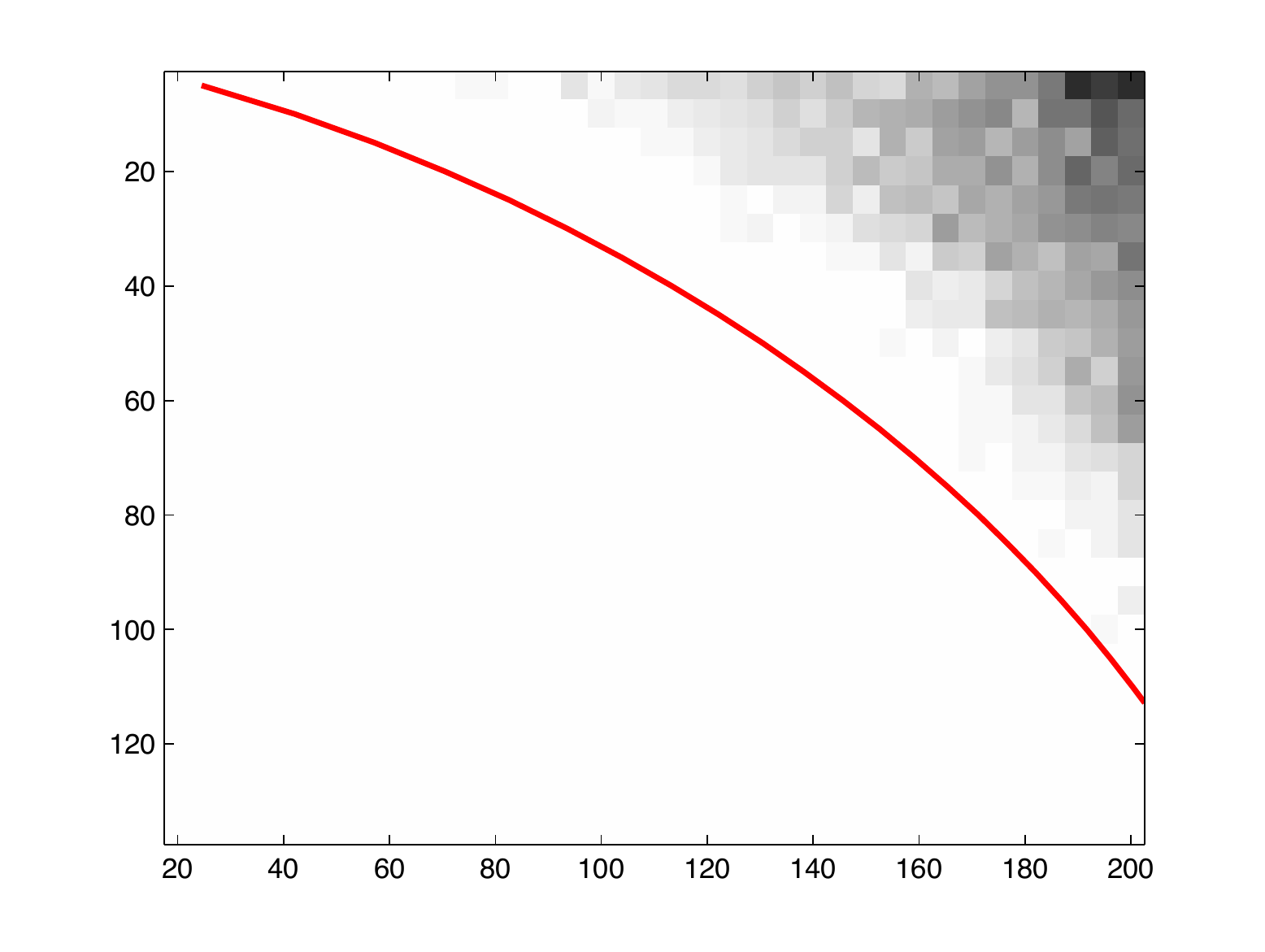}%
  \includegraphics[width=\figwidth]{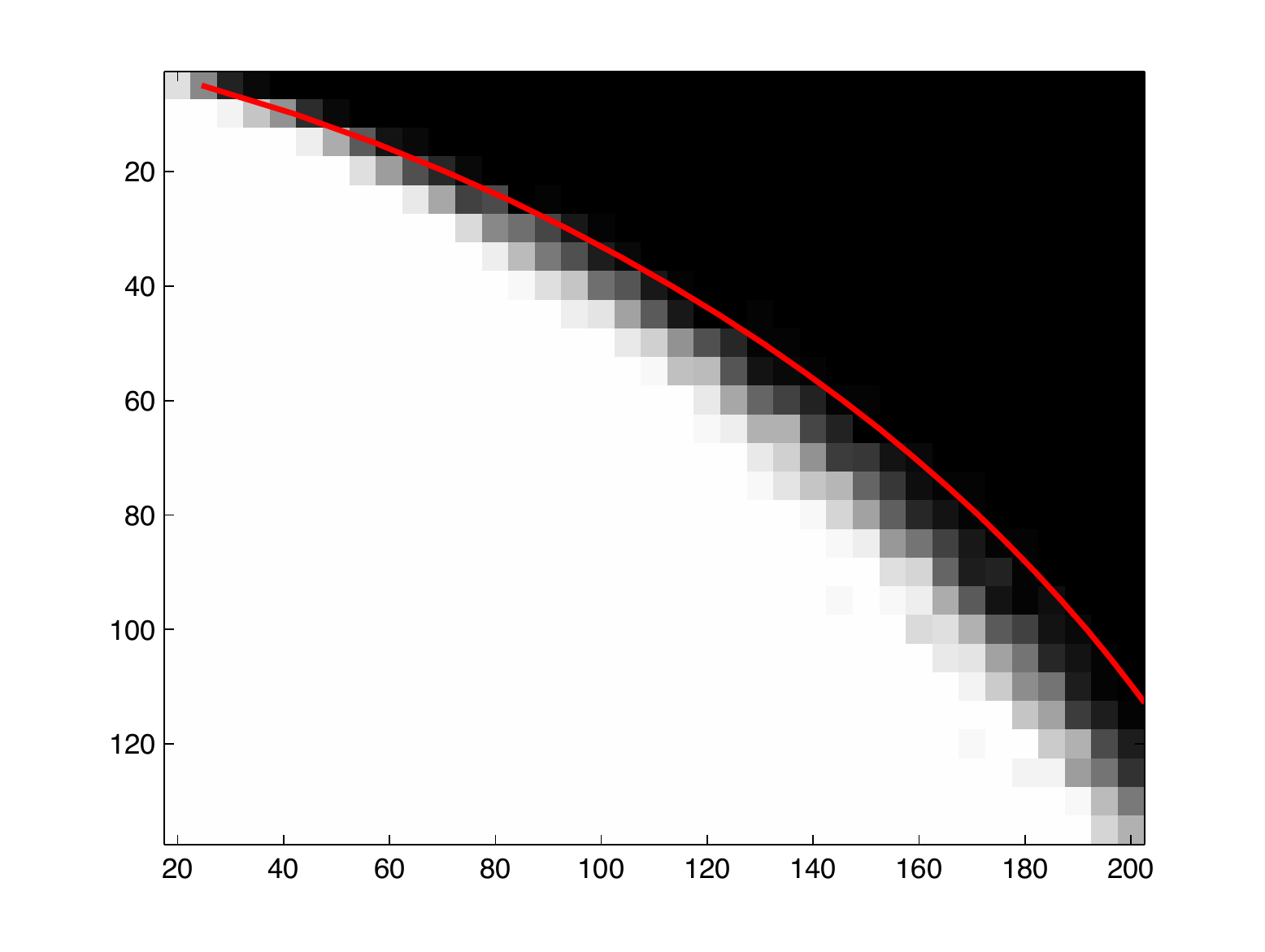}%
  \caption{Exact recovery counts (black means exact recovery) of
    basis-pursuit (left column) and weighted basis-pursuit (right
    column), where the $x$-axis is the sparsity ($s$) and the $y$-axis
    is the number of measurements ($m$). Exact recovery is declared
    with a tolerance equal to $10^{-5}$ on the first line, and equal
    to $10^{-6}$ on the second line. The red curve is a theoretical
    phase-transition threshold $s \mapsto s \log(e m / s)$}
  \label{fig:phase-cs}
\end{figure}

\subsection{A theoretical explanation}

Now, we want to understand if $\Delta_2$ can do better than
$\Delta_1$, and why. In particular, if $\Delta_1(Ax)$ is close to $x$
(but fails to reconstruct exactly $x$), under which condition do we
get $\Delta_2(Ax)=x$? In general, given a weight vector $w\in\R^N$,
what conditions on $w$ can insure that $\Delta_w(Ax)=x$? In
Theorem~\ref{thm:B} below, we use the duality argument
of~\cite{MR2236170} to prove that the condition
\begin{equation}
  \label{eq:A0}
  (A0)(I,C) \hspace{1cm} |w_{I^c}|_\infty\big|(1/w)_I\big|_2\leq C,
\end{equation} 
where $I$ is the support of $x$ and $C \geq 0$ is such that
\begin{equation*}
   C \leq \frac{1 - \delta}{\mu},
\end{equation*}
where $\delta$ and $\mu$ are, respectively, the restricted isometry
and incoherency constants~\cite{MR2300700,MR2236170,MR2243152} of the
matrix $A$, ensure  that the $w$-weighted algorithm $\Delta_w$
recovers exactly $x$ given
$Ax$.

It is interesting to note that, so far, only random matrices are able
to satisfy the incoherency and isometry properties for small values of
$m$. Thus, if one wants the number $m$ of measurements to be of the
order (up to some logarithmic factor) of the sparsity of the vector to
recover, one has to consider random matrices. This leads to results in
Compressed Sensing that hold with a large probability, with respect to
the randomness involved in the construction of the sensing matrix. In
practice, however, the most interesting sensing matrices are
structured matrices, like the Fourier or the Walsh matrices
(see~\cite{MR2300700,MR2417886}), since these matrices can be stored
and constructed by efficient algorithms. A lot of research go in this
direction, and we don't consider this problem here, but rather focus
on weighted algorithms. Therefore, we will state our probabilistic
results for a simple (and somehow universal) sensing matrix $A$ with
entries being i.i.d. centered Gaussian variables with variance $1/m$.
\begin{theorem}
  \label{thm:B}
  Let $x\in\R^N$ and denote by $I$ its support and by $s$ the
  cardinality of $I$. Let $C, \mu>0$ and $0 < \delta < 1$. Assume
  that
  \begin{equation*}
    m \geq c_0 \max\Big[\frac{s}{\delta^2}, \frac{s \log  N}{
      \mu^2}\Big] \;\mbox{ and }\;  C \leq \frac{1 - \delta}{\mu},
  \end{equation*}
  where $c_0$ is a purely numerical constant.  Consider the event
  $\Omega(I, C) = \{ |w_{I^c}|_\infty \big| (1/w)_I \big|_2 \leq C \}$
  and let $A$ be a $m\times N$ matrix with entries being
  i.i.d. centered Gaussian random variables with variance $1/m$. Then,
  with probability larger than
  \begin{equation*}
    1 - 2 \exp(-c_1 m \delta^2) - \exp\big(-c_2 \mu^2 m /s \big) - \P
    \Big[ \Omega(I, C)^\complement \Big],
  \end{equation*}
  the vector $x$ is exactly reconstructed by $\Delta_w(Ax)$.
\end{theorem}
Theorem~\ref{thm:B} gives an explicit condition, linking the
incoherency constant $\mu$, the restricted isometry constant $\delta$,
and the constant $C$ from condition $A0(I, C)$ on the weights $w$ that
ensures the exact reconstruction of $x$ using $\Delta_w$. This is the
first result of this nature for weighted basis pursuit.

When $w_{I^c}=0$ then $(A0)(I,C)$ holds with $C = 0$, so that one can
take $\delta = 1$ and $\mu = +\infty$. This is the case for $w =
(|\Delta_1(Ax)_i|)_{i=1}^N$ when $\Delta_1(Ax)=x$. This condition is
also satisfied when the weights vector $w$ is close enough to $|x|$
and when the absolute value of the non-zero coordinates of $|x|$ are
sufficiently large. For instance, $(A0)(I,C)$ holds when
\begin{equation}
  \label{eq:exemple-approx-weights}
  \min_{i\in I} |x_i| \geq \Big(1 + \frac{\sqrt{|I|}}{C}\Big) |w -
  |x||_\infty.
\end{equation}
Indeed, if we denote $\eps = |w - |x||_\infty$ then $(A0)(I,C)$
follows from~\eqref{eq:exemple-approx-weights} since $\max_{i\in
  I^c}w_i\leq \eps$ and
\begin{equation*}
  \Big| \Big(\frac{1}{w}\Big)_I \Big| \leq
  \frac{\sqrt{|I|}}{\min_{i\in I} 
    w_i}\leq \frac{\sqrt{|I|}}{\min_{i\in I}|x_i| - \eps}.
\end{equation*} 
In particular, if $A0(I, C)$ is satisfied with $C = c_1 / \sqrt{\log
  N}$, for some constant $0 < c_1 < 1$, then a proportional to $s$
number of Gaussian measurements will be enough to get $\Delta_w(Ax) =
x$ with a large probability.

In Figure~\ref{fig:A0-verif} below, we give an empirical illustration
of the fact that $A0(I, C)$ is indeed a relevant condition for exact
reconstruction of weighted basis-pursuit. We consider exactly the same
experiment as what we did in Section~\ref{sec:empirical-evidence}, but
this time we fix the number of measurements to $m = 110$ and the
sparsity of $x$ to $s = 45$. For this combination of $m$ and $s$, the
phase transition occurs, namely basis pursuit can either work or not,
see Figure~\ref{fig:phase-cs}, so we can expect for these values a
strong improvement of weighted basis-pursuit over non-weighted one. On
the left-side of Figure~\ref{fig:A0-verif}, we show the value of the
constant $C$ over the reweighting iterations. Namely, if $I$ is the
support of the true unknown vector $x$, we compute for $k = 1, \ldots,
K$ the values of
\begin{equation*}
  C^{k} = |w_{I^c}^{(k)}|_\infty\big|(1 / w^{(k)})_I\big|_2,
\end{equation*}
where
\begin{equation*}
  w^{(k)} = |\Delta_k^\eps(Ax)| + \eps
\end{equation*}
over the 10 repetitions (differentiated by different colors), where we
recall that $\Delta_k^\eps(Ax)$ is given by~\eqref{eq:delta-k-def} and
where we choose $K = 30$. On the right-side of
Figure~\ref{fig:A0-verif}, we show the logarithm of relative
reconstruction errors over the iterations, namely
\begin{equation*}
  \mathrm{err_k} = \log \Big( \frac{|\Delta_k^\eps(Ax) - x|_2}{|x|_2}
  \Big)
\end{equation*}
(we take the logarithm only for illustrational purpose, so that we can
see the cases when exact reconstructions occurs). Each repetition of
the experiment is represented with a different color. 

What we observe is a direct correspondence between the constant $C$
from Assumption $A0(I, C)$ and the quality of reconstruction of
weighted basis pursuit along the iterations. This tells that
Assumption $A0(I, C)$ indeed explains (at least in the considered
configuration) when exact reconstruction can or cannot happen using
weighted basis pursuit.

\setlength{\figlength}{6cm}%
\setlength{\figwidth}{7.1cm}

\begin{figure}[htbp]
  \centering
  \includegraphics[width=\figwidth,height=\figlength]{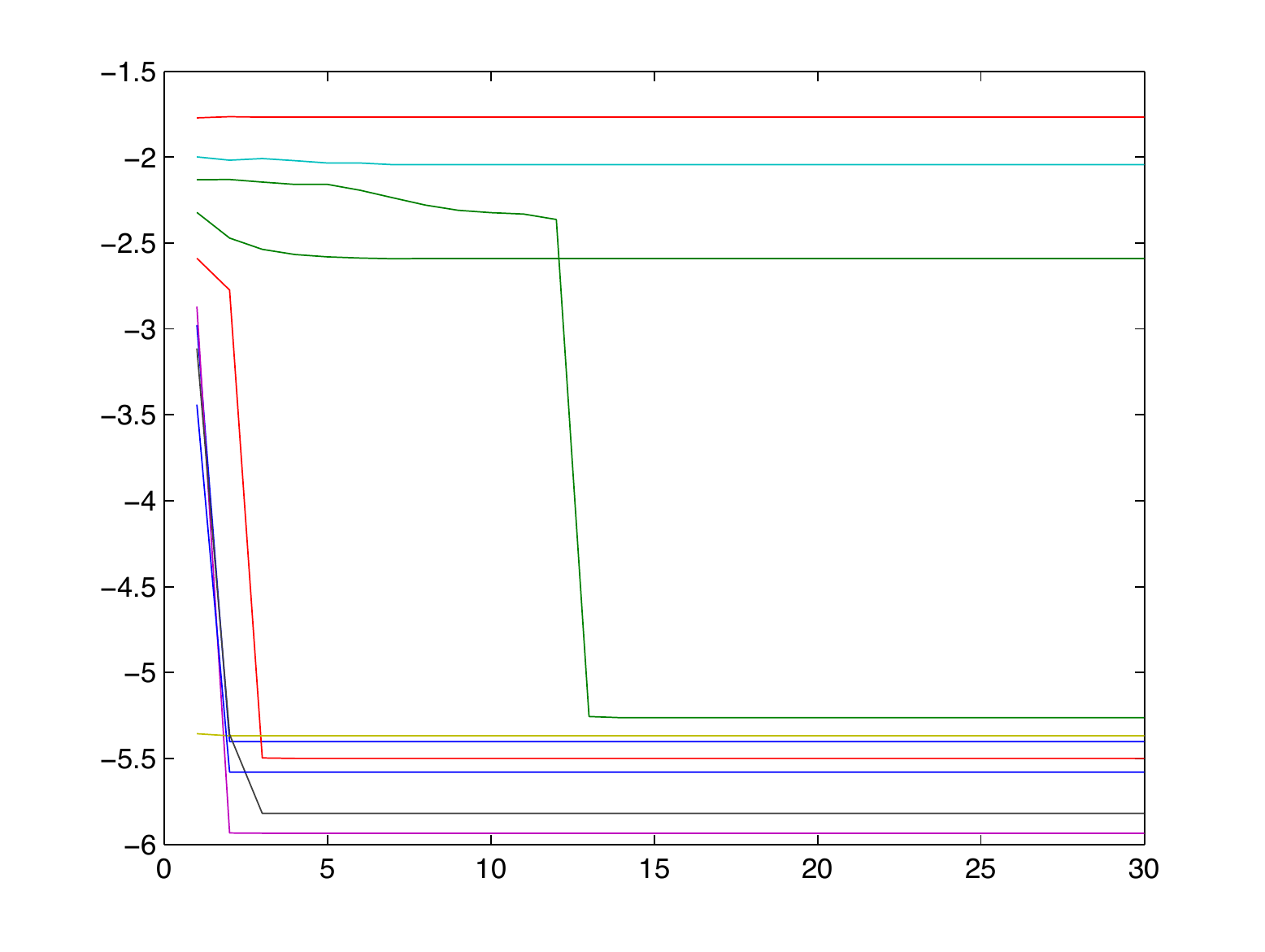}
  \hspace{-0.9cm}
  \includegraphics[width=\figwidth,height=\figlength]{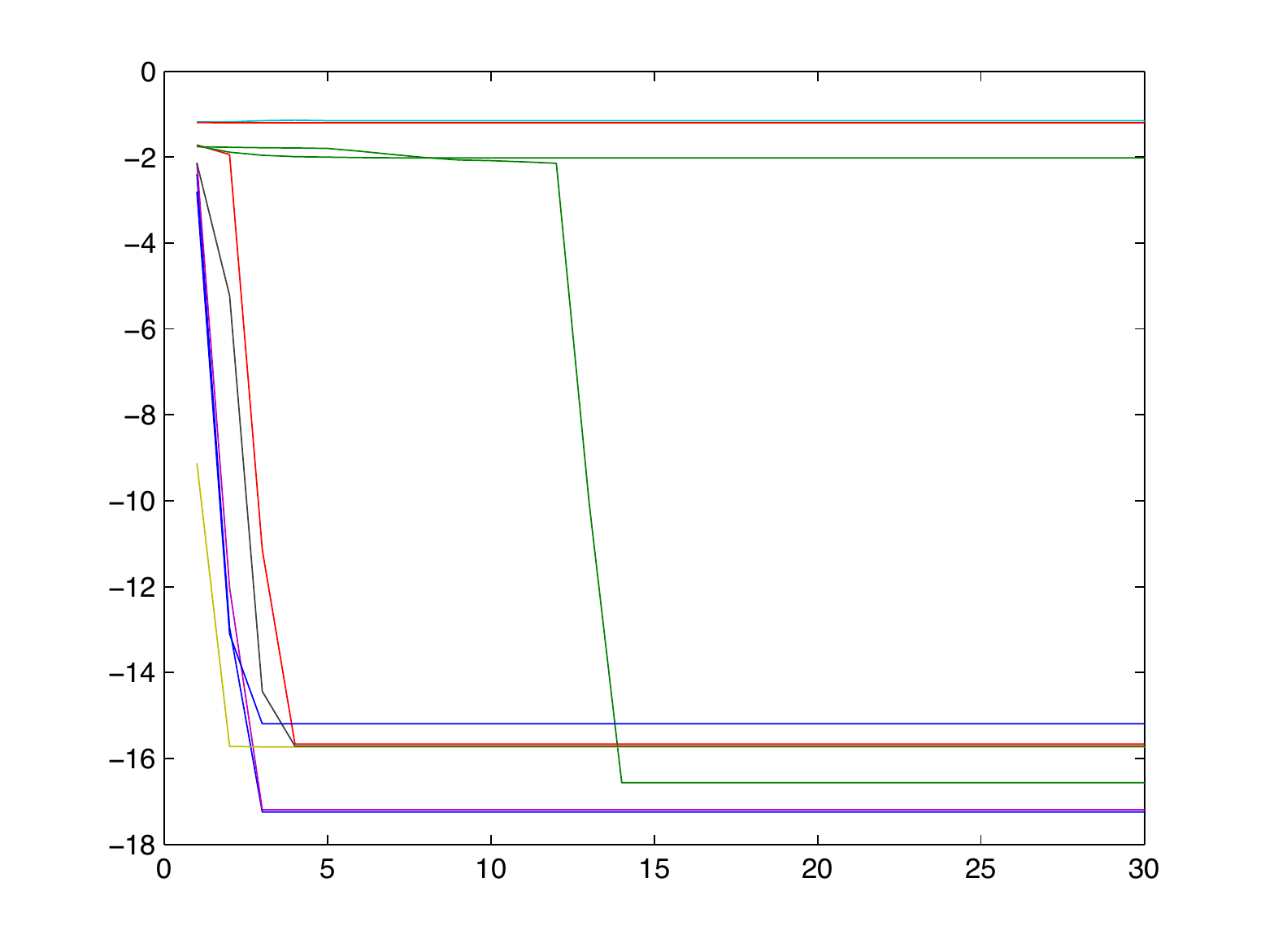}

  \caption{Logarithm of the value of the constant $C$ from
    Assumption~$A0(I, C)$ (left) and logarithm of the relative
    reconstruction error of weighted basis pursuit over the iterations
    (right).}
  \label{fig:A0-verif}
\end{figure}

\begin{remark}
  Note that uniform results can also be derived for the
  weighted-$\ell_1$ algorithm. Indeed, by using classical machinery,
  it can be proved that 1) implies 2) implies 3) where:
  \begin{enumerate}
  \item for all $x \in \Sigma_s$, $A \diag(w)$ satisfies
    $\mathrm{RIP}(\delta,8s)$ and $I_x\subset I_{w}$,
  \item $\sup_{x \in \ker(A \diag(w)) \cap
      B_1^N}\Norm{x}_2<\frac{1}{2\sqrt{s}}$ and $\forall x \in
    \Sigma_s, I_x\subset I_{w}$,
  \item for any $x\in \Sigma_s, \Delta_{w}(Ax)=x$.
  \end{enumerate}
  But, it is not clear why, for instance when $w = \Delta_1(Ax)$, it
  would be easier for the matrix $A \diag(\Delta_1(Ax))$ to satisfy
  $\mathrm{RIP}$ than for $A$ itself. The same remark also holds for
  the euclidean section of $B_1^N$ by the kernel of $A
  \diag(\Delta_1(Ax))$ or $A$. These approaches look too crude to
  perform a study of $\ell_1$-weighted algorithms, where most of the
  gain can be done only on the absolute multiplying constant in front
  of the minimal number of measurements $m$ needed for exact
  reconstruction.
\end{remark}

\subsection{Verifying exact reconstruction}

Thanks to Theorem~\ref{thm:A}, it is easy to test if we were able to
reconstruct exactly a vector $x$ given $Ax$. So far, we have to rely
on the theory to insure that with a high probability, we have
$\Delta_1(Ax) = x$. Using~\eqref{eq:weighted-BP-2}, we can verify this
belief. Indeed, Theorem~\ref{thm:A} entails that $\Delta_2(Ax) = x$
when $\Delta_1(Ax) = x$. In particular, if $\Delta_1(Ax)\neq
\Delta_2(Ax)$, then we are sure that we didn't perform the exact
reconstruction of $x$ using $\Delta_1(Ax)$. Then, we can iterate the
mechanism and define for any $k \geq 1$
\begin{equation*}
  \Delta_{k+1}(Ax) \in \argmin_{t \in \R^N}
  \Big(\sum_{i=1}^N \frac{|t_i|}{|\Delta_k(Ax)_i|} : At = Ax \Big),
\end{equation*}
leading to a sequence
\begin{equation}
  \label{eq:sequence-BP}
  \Delta_1(Ax), \Delta_2(Ax), \cdots, \Delta_r(Ax).
\end{equation}
If the sequence~\eqref{eq:sequence-BP} does not become constant after
a certain number of iterations, then it is very likely that none of
the algorithm $\Delta_k(Ax)$ reconstructed exactly $x$. We also have
the following reverse statement. Denote by $\Sigma_k$ the set of all
$k$-sparse vectors in $\R^N$.
\begin{theorem}
  \label{thm:C}
  Let $A$ be a $m \times N$ injective matrix on $\Sigma_m$ and let $x
  \in \Sigma_{\lfloor m / 2 \rfloor}$. The following statements are
  equivalent:
  \begin{enumerate}
  \item There exists an integer $r$ such that $\Delta_r(Ax)=x$,
  \item The sequence $\Delta_1(Ax), \Delta_2(Ax),\ldots,$ becomes
    constantly equal to a $\lfloor m/2 \rfloor$-sparse vector after a
    certain number of iterations.
  \end{enumerate}  
\end{theorem}

Note that the matrix with i.i.d. standard Gaussian entries is
injective on $\Sigma_m$ with probability one. Thus, we propose to
compute the sequence~\eqref{eq:sequence-BP} as an empirical test for
the exact reconstruction of a vector $x$ from $Ax$.

\section{Iteratively weighted soft-thresholding for matrix completion}
\label{cs-matrices}

In many applications, data can be represented as a database with
missing entries. The problem is then to fill the missing values of the
database, leading to the so-called~\emph{matrix completion}
problem. For instance, collaborative filtering aims at doing automatic
predictions of the taste of users, using the collected tastes of every
users at the same time~\cite{goldberg1992using}. The popular Netflix
prize is a popular application of this
problem\footnote{\texttt{http://www.netflixprize.com/}}. Other
applications include machine-learning~\cite{abernethy2006low},
control~\cite{554402}, quantum state
tomography~\cite{gross2010quantum}, structure from
motion~\cite{tomasi1992shape}, among many others. This problem can be
understood as a non-commutative extension of the compressed sensing
problem. So, a natural question is the following: \emph{Does the
  principle of iterative weighting of the $\ell_1$-norm work also for
  matrix completion?} In this Section, we prove empirically that the
answer to this question is yes. We prove that one can improve the
convex relaxation principle for matrices, which is based on the
nuclear norm \cite{MR2723472},~\cite{Gross}, by using a weighted
nuclear norm, in the same way as we did for vectors in
Section~\ref{sec:cs-vectors}. However, note that there is, as
explained below, a major difference between the vectors and matrices
cases at this point, since a weighted nuclear norm is not convex in
general, while a weighted $\ell_1$-norm is.

Let us first recall standard definitions and notations. Let $A_0 \in
\R^{n_1 \times n_2}$ be a matrix with $n_1$ rows and $n_2$
columns. The matrix $A_0$ is not fully observed. What we observe is a
given subset $\Omega \subset \{ 1, \ldots, n_1 \} \times \{ 1, \ldots,
n_2 \}$ of cardinality $m$ of the entries of $A_0$, where $m \ll n_1
n_2$. For any matrix $A \in \R^{n_1 \times n_2}$, we define the
\emph{masking} operator $\po (A) \in \R^{n_1 \times n_2}$ such that
$(\po (A))_{j,k} = A_{j,k}$ when $(j,k) \in \Omega$ and $(\po
(A))_{j,k} = 0$ when $(j,k) \notin \Omega$. We define also
$\po^\perp(A) = A - \po(A)$.

Since we consider the case where $m \ll n_1 n_2$, the matrix
completion problem is in general severely ill-posed. So, one needs to
impose a complexity or sparsity assumption on the unknown matrix
$A_0$. This is done by assuming that $A_0$ has low rank, which is the
natural extension of the sparsity assumption for vectors to the
spectrum of a matrix. For the problem of exact reconstruction, other
geometrical assumptions are necessary (such as the incoherency
assumption, see
\cite{candes-recht08,MR2723472,2009arXiv0901.2912A}). Under
such assumptions, it is now well-understood that the principle of
convex relaxation of the rank function is able to reconstruct exactly the
unknown matrix from few measurements, see
\cite{candes-recht08,MR2723472, Gross, recht2009simpler}. Indeed, a
natural approach would be to solve the problem
\begin{equation}
  \label{eq:rank-matrix-completion}
  \begin{split}
    &\text{ minimize } \rank{A} \\
    &\text{ subject to } \po(A) = \po(A_0),
  \end{split}
\end{equation}
but this minimization problem is known to be very hard to solve in
practice even for small matrices, see for instance
\cite{candes-recht08,MR2723472}. The convex envelope of the rank function over the unit ball of the operator norm is
the nuclear norm, see~\cite{fazel2002matrix}, which is given by
\begin{equation*}
  \norm{A}_1 = \sum_{j=1}^{n_1 \wedge n_2} \sigma_j(A),
\end{equation*}
(it is the bi-conjugate of the rank function over the unit ball of the
operator norm), where $\sigma_1(A) \geq \cdots \geq \sigma_{n_1 \wedge
  n_2}(A)$ are the singular values of $A$ in decreasing order. So, the
convex relaxation of~\eqref{eq:rank-matrix-completion} is
\begin{equation}
  \label{eq:nuclear-norm-matrix-completion}
  \begin{split}
    &\text{ minimize } \norm{A}_1 \\
    &\text{ subject to } \po(A) = \po(A_0).
  \end{split}
\end{equation}
This problem has received a lot of attention quite recently, see
\cite{candes-recht08,MR2723472, Gross,
  keshavan2009matrix,recht2009simpler}, among many others. The point
is that, in the same way as the basis pursuit for
vectors,~\eqref{eq:nuclear-norm-matrix-completion} is able to recover
exactly $A_0$ with a large probability, based on an almost minimal
number of samples (under some geometrical assumption).

In literature concerned about computational problems
\cite{springerlink:10.1007/s10107-009-0306-5},
\cite{mazumder_hastie_tibshirani09},
\cite{toh2009accelerated,liu2009implementable}, among others, the
relaxed version of~\eqref{eq:nuclear-norm-matrix-completion} is
considered, since it is easier to construct a solver for it (one can
apply generic first-order optimal methods, such as proximal
forward-backward splitting \cite{MR2203849}, among many other methods)
and since it is more stable in the presence of noise. Note that the
SVT algorithm of~\cite{MR2600248} gives a solution under equality
constraints for an objective function with an extra ridge term
$\norm{A}_1 + \tau \norm{A}_2^2$.  The relaxed problem is simply
formulated as penalized least-squares:
\begin{equation}
  \label{eq:matrix-lasso}
  \hat A_\lambda \in \argmin_{A \in \R^{n_1 \times n_2}} \Big\{
  \frac 12 \norm{\cP_\Omega(A) - \cP_\Omega(A_0)}_2^2 + \lambda
  \norm{A}_1 \Big\},
\end{equation}
where $\lambda > 0$ is a parameter balancing goodness-of-fit and
complexity, measured by the nuclear norm.

Before we go on, we need some notations. The vector of singular values of
$A$ is denoted by $\sigma(A) = (\sigma_1(A), \ldots, \sigma_r(A))$,
sorted in non-increasing order, where $r$ is the rank of $A$.  We
define, for $p \geq 1$, the $p$-Schatten norm by
\begin{equation*}
  \norm{A}_p = |\sigma(A)|_p,
\end{equation*}
which is the $\ell_p$ norm of $\sigma(A)$. We shall denote also by
$\norm{A} = \norm{A}_\infty = \sigma_1(A)$ the operator norm of $A$,
and note that $\norm{A}_2$ is the Frobenius norm, associated to the
Euclidean inner product $\inr{A, B} = \tr(A^\top B)$, where $\tr(A)$
stands for the trace of $A$. For any matrix $A$ its singular values
decomposition (SVD) writes as $A = U \diag(\sigma(X)) V^\top$, where
$\diag(\sigma(X))$ is the diagonal matrix with $\sigma(A)$ on its
diagonal, and $U$ and $V$ are, respectively $n_1 \times r$ and $n_2
\times r$ orthonormal matrices.

\subsection{A new algorithm for matrix completion}

We have in mind to do the same as we did in
Section~\ref{sec:cs-vectors} for the reconstruction of sparse
vectors. For a given weight vector $w = (w_1, \ldots, w_{n_1 \wedge
  n_2})$, with $w_1 \geq \cdots \geq w_{n_1 \wedge n_2} \geq 0$, we
consider
\begin{equation}
  \label{eq:weighted-matrix-lasso}
  \tilde A_{\lambda}^w \in \argmin_{A \in \R^{n_1 \times n_2}} \Big\{
  \frac 12 \norm{\cP_\Omega(A) - \cP_\Omega(A_0)}_2^2 + \lambda
  \norm{A}_{1, w} \Big\},
\end{equation}
where $\norm{A}_{1, w}$ is the weighted nuclear-norm
\begin{equation}\label{eq:weighted-S1-norm}
  \norm{A}_{1, w} = \sum_{j=1}^{n_1 \wedge n_2} \frac{\sigma_j(A)}{w_j},
\end{equation}
with the convention $1/0 = +\infty$. Now, we would like to use the
idea of reweighting using previous estimates, in the same as we did in
Section~\ref{sec:cs-vectors}: if $\hat A_\lambda$ is a solution
to~\eqref{eq:matrix-lasso}, we want to use for instance
\begin{equation*}
  w_j = \sigma_j(\hat A_\lambda),  
\end{equation*}
and find a solution to the problem~\eqref{eq:weighted-matrix-lasso}
for this choice of weights. But, let us stress the fact that, while we
call $\norm{\cdot}_{1,w}$ the weighted nuclear norm, it is not a norm,
since it is not a convex function in general!  A simple
counter-example is as follows. If $w_1 > w_2$ (which is usually the case since
singular values are taken in a non-increasing order) then for $A =
\diag(1,0,\ldots,0)$ and $B = \diag(0,1,0,\ldots,0)$, we have
\begin{equation*}
  \frac{\norm{A}_{1,w} + \norm{B}_{1,w}}{2} = \frac{s_1(A) +
    s_1(B)}{2w_1} = \frac{1}{w_1} < \frac{1}{2} \Big(\frac{1}{w_1} +
  \frac{1}{w_2}\Big) = \Big\|\frac{A+B}{2} \Big\|_{1,w},
\end{equation*}
hence $\norm{\cdot}_{1,w}$ is not convex. Moreover, since the aim of
$\norm{\cdot}_{1, w}$ is to promote low-rank matrices, the weight
vector $w$ should be chosen non-increasing, corresponding precisely to
the case where $\norm{\cdot}_{1, w}$ is non-convex (note that when $0
< w_1 \leq w_2 \leq \cdots\leq w_{n_1\wedge n_2}$, it is easy to prove
that $\norm{\cdot}_{1,w}$ is a norm).
Consequently,~\eqref{eq:weighted-matrix-lasso} is not a convex
minimization problem in general, and a minimization algorithm is very
likely to be stuck at a local minimum. But we would like to stick to
the idea of reweighting, since it worked well for CS.

The first idea that may come to mind is to use a convex relaxation of
the non-convex function $\norm{\cdot}_{1, w}$ (just as convex
relaxation of the rank function led to the nuclear norm), but it
simply leads back to the nuclear norm itself!  Indeed, it can be
proved that if $w_1 \geq w_2 \geq \cdots \geq w_{n_1\wedge n_2}>0$,
the convex envelope of $\norm{\cdot}_{1,w}$ on the ball $\{ A :
\norm{A}_1 \leq 1 \}$ is simply $A \mapsto \norm{A}_{1} / w_1$.

Let us go back to the original problem~\eqref{eq:matrix-lasso}. It
turns out that~\eqref{eq:matrix-lasso} is equivalent to the fact that
$\hat A_\lambda$ satisfies the following fixed-point equation:
\begin{equation}
  \label{eq:fixed-point-matrix-lasso}
  \hat A_\lambda = S_\lambda( \cP_{\Omega}^\perp(\hat A_\lambda) +
  \cP_\Omega(A_0)),
\end{equation}
where $S_\lambda$ is the spectral soft-thresholding operator defined
for every $B\in \R^{n_1\times n_2}$ by
\begin{equation*}
  S_\lambda(B) =  U_B \diag\Big( ( \sigma_1(B) -
  \lambda)_+, \ldots, (\sigma_{\rank(B)}(B) -
  \lambda)_+ \Big) V_B^\top,
\end{equation*}
where $B = U_B \Sigma_B V_B^\top$ is the SVD of $B$, with $\Sigma_B =
\diag(\sigma_1(B), \ldots, \sigma_{\ra(B)}(B))$. This fact is easily
explained. Indeed, define $f_2(A) = \frac 12 \norm{\po (A) - \po
  (A_0)}_2^2$, which is a differentiable function with gradient $\grad
f_2(A) = \po (A) - \po (A_0)$ and $f_1(A) = \lambda \norm{A}_1$, which
is a non-differentiable convex function. We will denote by $\partial
f_1(A)$ the subdifferential of $f_1$ at $A$. The fact that $\hat
A_\lambda \in \argmin_{A} \{ f_2(A) + f_1(A) \}$ is equivalent to the
fact that $0 \in \partial(f_1 + f_2)(\hat A_\lambda) = \{ \grad
f_2(\hat A_\lambda) \} + \partial f_1(\hat A_\lambda)$ (for the
Minkowskii's addition of sets), that we rewrite in the following way:
\begin{equation}
  \label{eq:matrix-lasso-charac}
  \hat A_\lambda - \grad f_2(\hat A_\lambda) - \hat A_\lambda
  \in \partial f_1(\hat A_\lambda).
\end{equation}
On the other hand, a standard tool in convex analysis is the
\emph{proximal} operator, \cite{MR2203849},~\cite{MR0274683}. The
proximal operator of a convex function, for instance $f_1$, is given,
for every $B\in \R^{n_1\times n_2}$, by
\begin{equation*}
  \prox_{f_1}(B) = \argmin_{A \in \R^{n_1 \times n_2}} \Big\{
  \frac{1}{2} \norm{A - B}_2^2 + f_1(A)  \Big\},
\end{equation*}
the minimizer being unique since $A \mapsto \frac 12 \norm{A - B}_2^2
+ f_1(A)$ is strongly convex. But, since $\partial(\frac 12
\norm{\cdot - B}_2^2 + f_1(\cdot))(A) = \{ A - B \} + \partial
f_1(A)$, the point $\prox_{f_1}(B)$ is uniquely determined by the
inclusion
\begin{equation}
  \label{eq:prox-charac}
  B - \prox_{f_1}(B) \in \partial f_1(\prox_{f_1}(B)).
\end{equation}
So, choosing $B = \hat A_\lambda - \grad f_2(\hat A_\lambda)$
in~\eqref{eq:prox-charac} and identifying
with~\eqref{eq:matrix-lasso-charac} leads to the fact that $\hat
A_\lambda$ satisfies the fixed-point equation
\begin{equation*}
  \hat A_\lambda = \prox_{f_1}(\hat A_\lambda - \grad f_2(\hat
  A_\lambda)),
\end{equation*}
which leads to~\eqref{eq:fixed-point-matrix-lasso} on this particular
case, since we know that $\prox_{f_1}(B) = S_\lambda(B)$ (see
Proposition~\ref{prop:weighted-nuclear-prox} below). Note that the
same argument proves that, if we add a ridge term to the nuclear norm
penalization, namely
\begin{equation}
  \label{eq:matrix-enet}
  \hat A_{\lambda, \tau} = \argmin_{A \in \R^{n_1 \times n_2}} \Big\{
  \norm{\cP_\Omega(A) - \cP_\Omega(A_0)}_2^2 + 2 \lambda \norm{A}_1 +
  \tau \norm{A}_2^2 \Big\}
\end{equation}
for any $\tau \geq 0$, then and equivalent formulation is the fixed
point equation
\begin{equation}
  \label{eq:fixed-point-enet}
    \hat A_{\lambda, \tau} = \frac{1}{1 + \tau} S_\lambda(
    \cP_{\Omega}^\perp(\hat A_{\lambda, \tau}) + \cP_\Omega(A_0)),
\end{equation}
and the minimizer is unique this time, since the objective function is
now strongly convex.

The argument given above is at the core of the proximal operator
theory, and leads to the so-called proximal forward-backward splitting
algorithms, see~\cite{MR2203849,MR701288} and~\cite{MR2486527}. Since
these algorithm are optimal among the class of first-order algorithms,
they drawn a large attention in the machine learning community, see
for instance the survey~\cite{bach-book-chapter}. Another advantage in
the case of matrix completion is that such an algorithm can handle
large scale matrices, see Remark~\ref{rem:large-scale} below.

So, we have seen that~\eqref{eq:matrix-lasso}
and~\eqref{eq:fixed-point-matrix-lasso}, or~\eqref{eq:matrix-enet}
and~\eqref{eq:fixed-point-enet} are equivalent formulations of the
same problem. So, instead of
considering~\eqref{eq:weighted-matrix-lasso}, we could consider the
corresponding fixed-point problem. Unfortunately, since
$\norm{\cdot}_{1, w}$ is non-convex, the above arguments based on the
subdifferential does not make sense anymore. But still, we can
consider an estimator defined as a fixed point equation for the
weighted soft-thresholding operator.
\begin{theorem}
  \label{prop:existence-and-unicity}
  Assume that $\tau > 0$ and $w_1 \geq \cdots \geq w_{n_1 \wedge n_2}
  \geq 0$. Let us define the matrix $\hat A_{\lambda}^w$ as the
  solution of the fixed-point equation
  \begin{equation}
    \label{eq:weighted-fixed-point}
    \hat A_{\lambda}^w = \frac{1}{1 + \tau} S_{\lambda}^w(
    \cP_{\Omega}^\perp(\hat A_{\lambda}^w) + \cP_\Omega(A_0)),
  \end{equation}
  where $S_{\lambda}^w$ is the weighted soft-thresholding operator
  given by
  \begin{equation}
    \label{eq:s-lambda-w-def}
    S_{\lambda}^w(B) =  U_B  \diag\Big( \Big( \sigma_1(B) -
    \frac{\lambda}{w_1} \Big)_+, \ldots, \Big(\sigma_{\ra(B)}(B) -
    \frac{\lambda}{w_{\ra(B)}}\Big)_+ \Big) V_B^\top,
  \end{equation}
  where $B = U_B \diag(\sigma(B)) V_B^\top$ is the SVD of $B$. Then,
  the solution to~\eqref{eq:weighted-fixed-point} exists and is
  unique.
\end{theorem}
Theorem~\ref{prop:existence-and-unicity} is proved in
Section~\ref{sec:proofs-cs-matrices} below, and is a by-product of our
analysis of the iterative scheme to approximate the solution
of~\eqref{eq:weighted-fixed-point}. The parameter $\tau > 0$ can be
arbitrarily small (in our numerical experiments we take it equal to
zero, see Section~\ref{sec:mc-numerical-study}), but it ensures
unicity and convergence of the iterative scheme proposed below. Once
again, let us stress the fact that~\eqref{eq:weighted-fixed-point}
(with $\tau=0$) is not equivalent to~\eqref{eq:weighted-matrix-lasso}
in general, since $A \mapsto \norm{A}_{1, w}$ is not convex.

The consideration of~\eqref{eq:weighted-fixed-point} has several
advantages: we guarantee unicity of the solution, while the
problem~\eqref{eq:weighted-matrix-lasso} may have several solutions,
and it is easy to solve the fixed-point
problem~\eqref{eq:weighted-fixed-point} using iterations. Even
further, from a numerical point of view, it can be easily used
together with a continuation algorithm, as explained in
Section~\ref{sec:mc-numerical-study} below, to compute a set of
solutions for several values of the smoothing parameter $\lambda$.

The next Theorem proves that iterates of the fixed-point
Equation~\eqref{eq:weighted-fixed-point} converges exponentially fast
to the solution.
\begin{theorem}
  \label{thm:algorithm_convergence}
  Take $A^0$ as the matrix with zero entries and define for any $k
  \geq 0$:
  \begin{equation}
    \label{eq:iterations}
    A^{k+1} = \displaystyle \frac{1}{1 + \tau}
    S_\lambda^{w}(\cP_{\Omega}^\perp(A^k) + \cP_\Omega(A_0)).
  \end{equation}
  Then, for any $n \geq 1$, one has:
  \begin{equation*}
    \| \hat A_\lambda^w - A^n \|_2 \leq \frac{1}{\tau (1 +
      \tau)^{n}} \| \po(A_0) \|_2,
  \end{equation*}
  where $\hat A_\lambda^w$ is the solution
  of~\eqref{eq:weighted-fixed-point}.
\end{theorem}

The proof of Theorem~\ref{thm:algorithm_convergence} is given in
Section~\ref{sec:proofs-cs-matrices}. The main step of the proof is to
establish the Lipshitz property of the weighted soft-thresholding
operator, see Proposition~\ref{prop:S-w-lipshitz}. Since $S_\lambda^w$
is not a proximal operator (the objective function is not convex), we
cannot use directly the property of firm-expansivity, which is a
direct consequence of the definition of a proximal operator, see the
discussion in Section~\ref{sec:proofs-cs-matrices}.

\subsection{Numerical study}
\label{sec:mc-numerical-study}

\subsubsection{Algorithms}

In this Section we compare empirically the quality of reconstruction
using nuclear norm minimization~\eqref{eq:matrix-lasso} (NNM), or
equivalently~\eqref{eq:fixed-point-matrix-lasso}, and weighted
spectral soft-thresholding~\eqref{eq:weighted-fixed-point} (WSST). To
compute the NNM we use the Accelerated Proximal Gradient (APG)
algorithm of~\cite{toh2009accelerated} using the \texttt{MATLAB}
package~\texttt{NNLS}, which is a state-of-the-art solver for the
minimization problem~\eqref{eq:matrix-lasso}. This algorithm is based
on an accelerated proximal gradient algorithm, itself based on the
accelerated gradient of Nesterov,
see~\cite{MR701288,nesterov2007gradient} and the FISTA algorithm, see
\cite{MR2486527} and see also~\cite{Ji:2009:AGM:1553374.1553434} for a
similar algorithm. In the APG algorithm, we use the linesearch and the
continuation techniques, see \cite{toh2009accelerated}, but we don't
use truncation, since it led to poor results in the problems
considered here. The target value of $\lambda$ for NNM and WSST
(see~\eqref{eq:matrix-lasso} and~\eqref{eq:weighted-fixed-point}) is
simply taken as $\lambda_{\mathrm{target}} = \eps \times \norm{\po(
  A_0)}_\infty$, with $\eps = 10^{-4}$ or $\eps = 10^{-3}$ depending
on the problem, see below. The solution coming out of the APG
algorithm is denoted by $\hat A_\lambda^{(0)}$. Note that we could
have used the FPC~\cite{springerlink:10.1007/s10107-009-0306-5} or
SVT~\cite{MR2600248} algorithms instead, but it led in our experiments
to poorer results compared to the APG (in particular when looking for
solutions with a rank of order, say, 100 on ``real'' matrices, like in
the inpainting or recommanding systems, see below).

The WSST is computed following the Algorithm~\ref{alg:WSST} below. The
first while loop is a continuation loop, that goes progressively to
$\lambda_{\mathrm{target}}$. Doing this instead of using
$\lambda_{\mathrm{target}}$ directly is known to improve stability and
rate of convergence of the algorithm. It does not take more time than
using $\lambda_{\mathrm{target}}$ directly (actually, it usually takes
less time), since we use warm starts: when taking a smaller $\lambda$,
we use the previous value $A_{\mathrm{new}}$ (the solution with the
previous $\lambda$) as a starting point.  Once we reached
$\lambda_{\mathrm{target}}$, we obtain a first solution of the fixed
point problem~\eqref{eq:weighted-fixed-point}, denoted by $\hat
A_\lambda^{(1)}$. Then, we update the weights by taking $w_j =
\sigma_j(\hat A_\lambda^{(1)})$, and we start all over. We don't use a
continuation loop again, since we are already at the desired value of
$\lambda$. We keep the parameter $\lambda$ fixed, we only repeat the
process of updating the weights and finding the solution to the fixed
point~\eqref{eq:weighted-fixed-point} $K$ times. By doing this, we are
typically going to decrease (eventually a lot) the final rank of the
WSST, while keeping a good reconstruction accuracy. This process of
updating the weights is usually not long. Typically, after a small
number of iterations, two fixed-point solutions before and after an
update are very close, so that our choice $K = 50$ is typically too
large, but we keep it this way to ensure a good stability of the final
solution.

Note that in Algorithm~\ref{alg:WSST} we use the
iterations~\eqref{eq:iterations} with $\tau = 0$, since it gives
satisfactory results. We use a simple stopping rule
$\norm{A_{\mathrm{new}} - A_{\mathrm{old}}}_2 /
\norm{A_{\mathrm{old}}}_2 \leq \text{tol}$ with $\text{tol} = 5 \times
10^{-4}$ or $\text{tol} = 10^{-3}$ depending on the scaling of the
problem, see below. We used in all our computations $q = 0.7$ and $K =
50$.  For a fair comparison, we always use, for a reconstruction
problem, the same parameters $\eps, \mathrm{tol}$ and $\lambda$ for
both NNM and WSST. Of course, for the WSST we need to rescale
$\lambda$ by multiplying it by $w_1$ (the first coordinate of the
weights vector, which is equal to $\sigma_1(\hat A^{(0)})$ at the
first iteration).

\begin{remark}
  \label{rem:large-scale}
  A good point with WSST is that it can handle large scale matrices,
  since at each iteration one only needs to store $A_{\mathrm{old}}$,
  which is a low rank matrix (coming out of a previous spectral
  soft-thresholding) and $\po(A_{\mathrm{old}} + A_0)$, which is a
  sparse matrix.
\end{remark}

\begin{remark}
  The overall computational cost of WSST is obviously much longer than
  the one of NNM, since we use $K$ iterations, and since we don't use
  accelerated gradient, linesearch and other accelerating recipes in
  our implementation of WSST. This is done purposely: we want to
  compare the quality of reconstruction of the ``pure'' WSST, without
  helping computational tricks, that usually improves rate of
  convergence, but accuracy of reconstruction as well (this is the
  case if one compares NNM with and without these tools).
\end{remark}

\begin{algorithm}[htbp]
  \small
  \KwIn{The observed entries $\po(A_0)$, a preliminary reconstruction
    $\hat A_\lambda^{(0)}$ and parameters $\lambda_1 >
    \lambda_{\text{target}} > 0$, $0 < q, \mathrm{tol} < 1$, $K \geq
    1$}
  
  \KwOut{The WSST reconstruction $\hat A_\lambda^{(K)}$}%
  
  Put $A_{\text{new}} = 0$, $\lambda = \lambda_1$ and take $w_j =
  \sigma_j(\hat A_\lambda^{(0)})$

  \While{$\lambda > \lambda_{\mathrm{target}}$}{

    Put $\delta = +\infty$

    \While{$\delta > \mathrm{tol}$}{
      
      $A_{\text{old}} = A_{\text{new}}$
      
      $A_{\text{new}} = S_\lambda^w (A_{\text{old}} -
      \po(A_{\text{old}}) + \po(A_0) )$ 

      $\delta = \norm{A_{\mathrm{new}} - A_{\mathrm{old}}}_2 /
      \norm{A_{\mathrm{old}}}_2$ }
        
    $\lambda = \lambda \times q$           
  }
  
  Put $\hat A_\lambda^{(1)} = A_{\mathrm{new}}$
  
  \For{$k=1, \ldots, K$}{
    
    Put $w_j = \sigma_j(\hat A_\lambda^{(k)})$ and $\delta = +\infty$
    
    \While{$\delta > \mathrm{tol}$}{
      
      $A_{\text{old}} = A_{\text{new}}$
      
      $A_{\text{new}} = S_\lambda^w (A_{\text{old}} -
      \po(A_{\text{old}}) + \po(A_0) )$ 
     
      $\delta = \norm{A_{\mathrm{new}} - A_{\mathrm{old}}}_2 /
      \norm{A_{\mathrm{old}}}_2$        
    }
                
  }

  \Return $\hat A_\lambda^{(K)}$

  \caption{Computation of the iteratively weighted spectral
    soft-thresholding.}
  \label{alg:WSST}
\end{algorithm}

\subsubsection{Phase transition}

In Figure~\ref{fig:phase-transision}, we give a first empirical
evidence of the fact that WSST improves a lot upon NNM. For each $r
\in \{ 5, 10, 15, \ldots, 80 \}$, we repeat the following experiment
50 times. We draw at random $U$ and $V$ as $500 \times r$ matrices
with $N(0, 1)$ i.i.d entries, and put $A_0 = U V^\top$ (which is rank
$r$ a.s.). Then, we choose uniformly at random $30\%$ of the entries
of $A_0$, and compute the NNM and the WSST based on this matrix. In
Figure~\ref{fig:phase-transision}, we show, for each $r$ (x-axis), the
boxplots of the relative reconstruction errors $\norm{\hat A - A_0}_2
/ \norm{A_0}_2$ over the 50 repetitions for $\hat A =$ NNM (top-left)
and $\hat A$ = WSST (top-right). On this example, we observe that NNM
is not able to recover matrices with a rank larger than 35, while WSST
can recover matrices with a rank up to 70. The boxplots of the ranks
recovered by NNM and WSST are on the second line, where we observe
that WSST always recovers the true rank up to a rank of order $70$,
while NNM correctly recovers the rank (only most of the time) up to a
rank $35$, and overestimates it a lot for larger ranks. So, on this
simulated example, we observe a serious improvement of NNM using WSST,
since the latter has the exact reconstruction property for matrices
with twice a larger rank ($70$ instead of $35$).

\setlength{\figlength}{7.1cm}%
\setlength{\figwidth}{7.1cm}

\begin{figure}[htbp]
  \centering
  \includegraphics[width=\figwidth,height=\figlength]{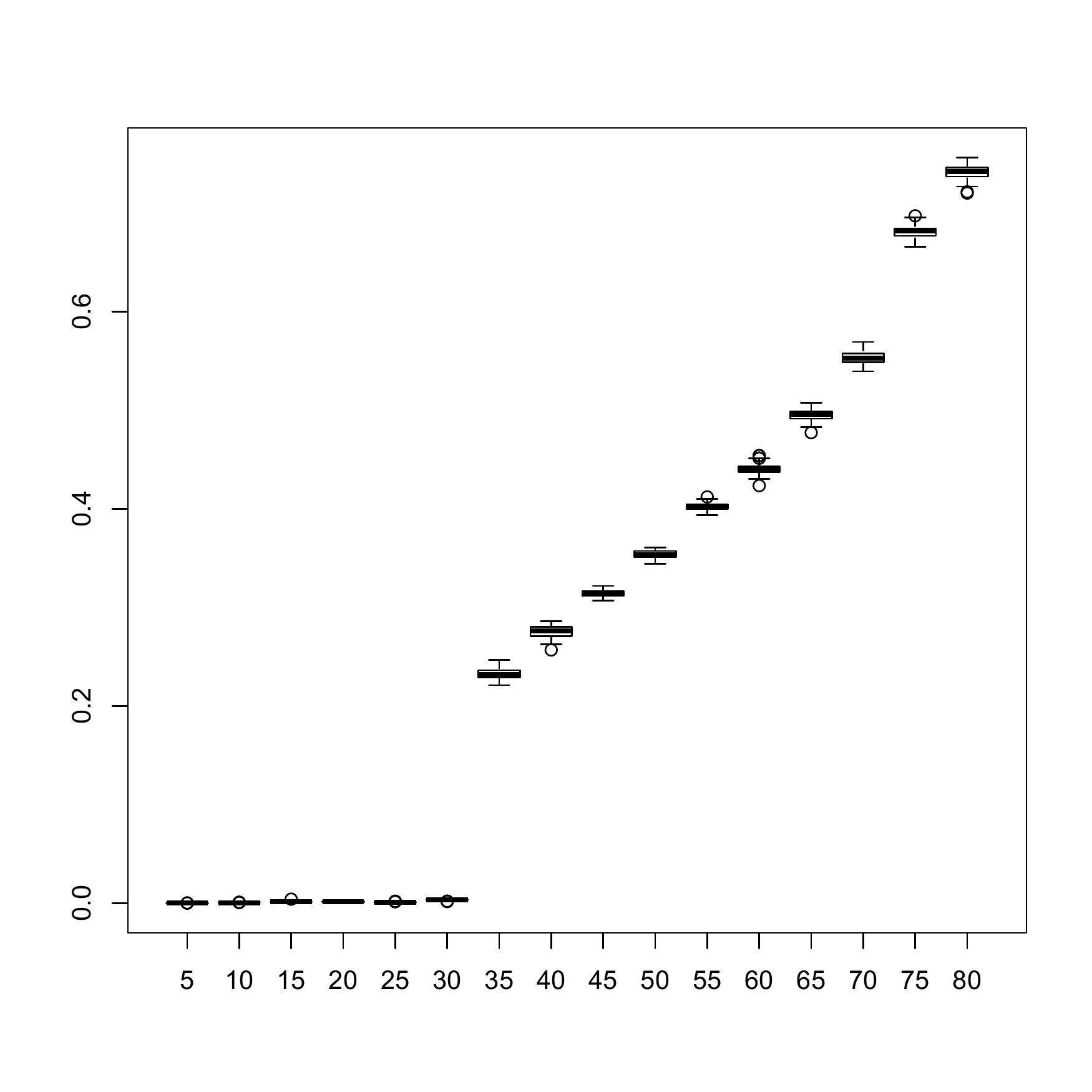}
  \hspace{-0.9cm}
  \includegraphics[width=\figwidth,height=\figlength]{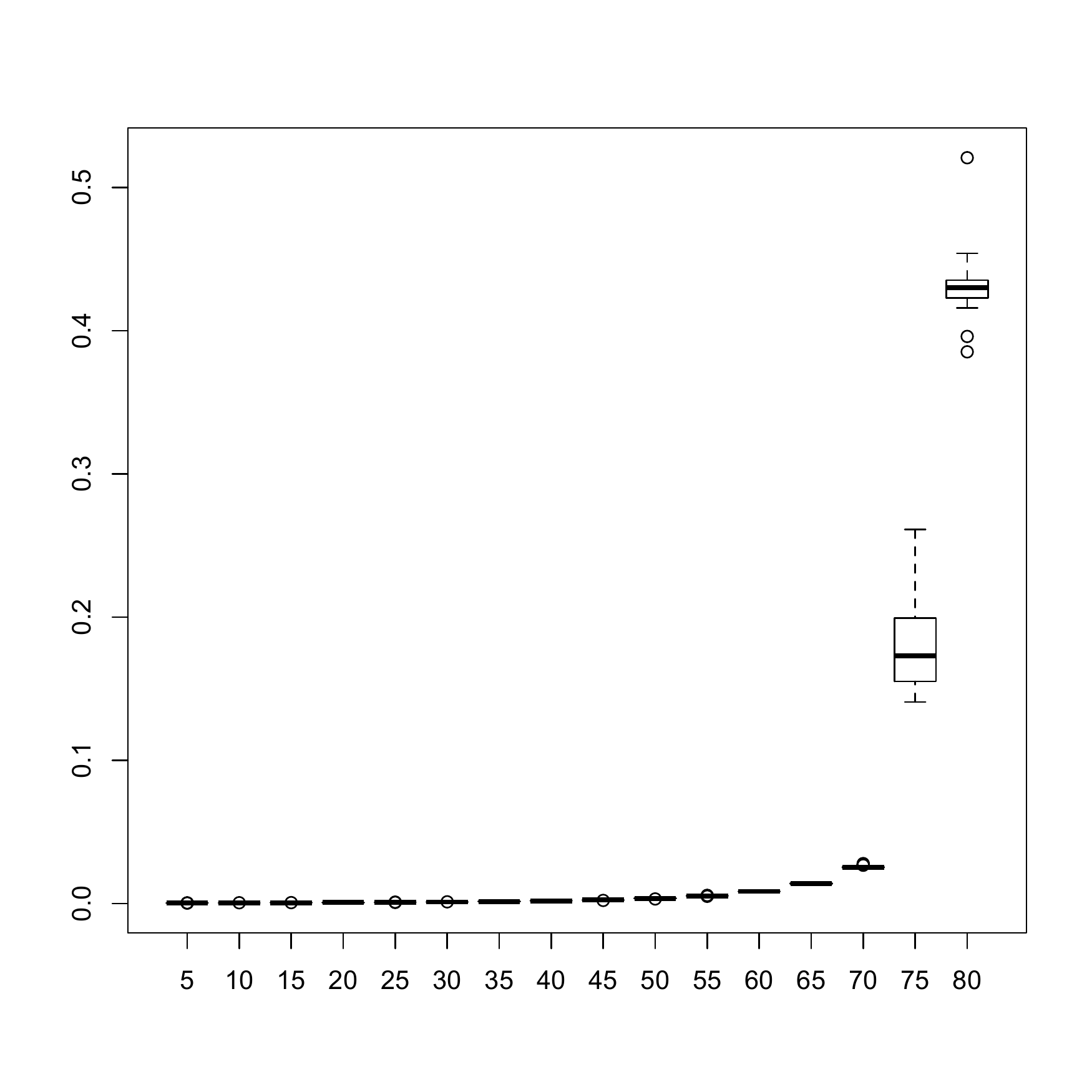} \\
  \vspace{-1cm}
  \includegraphics[width=\figwidth,height=\figlength]{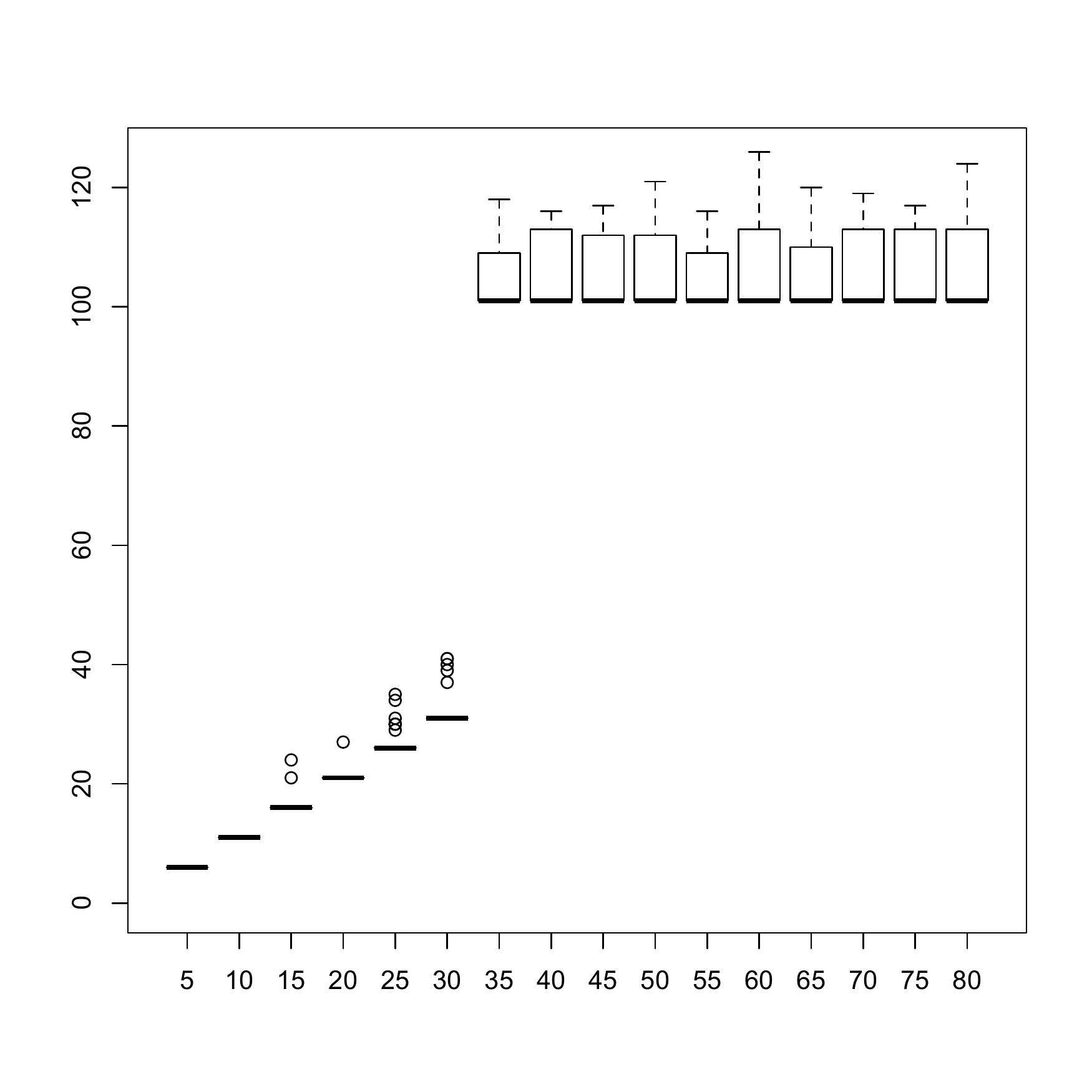}
  \hspace{-0.9cm}
  \includegraphics[width=\figwidth,height=\figlength]{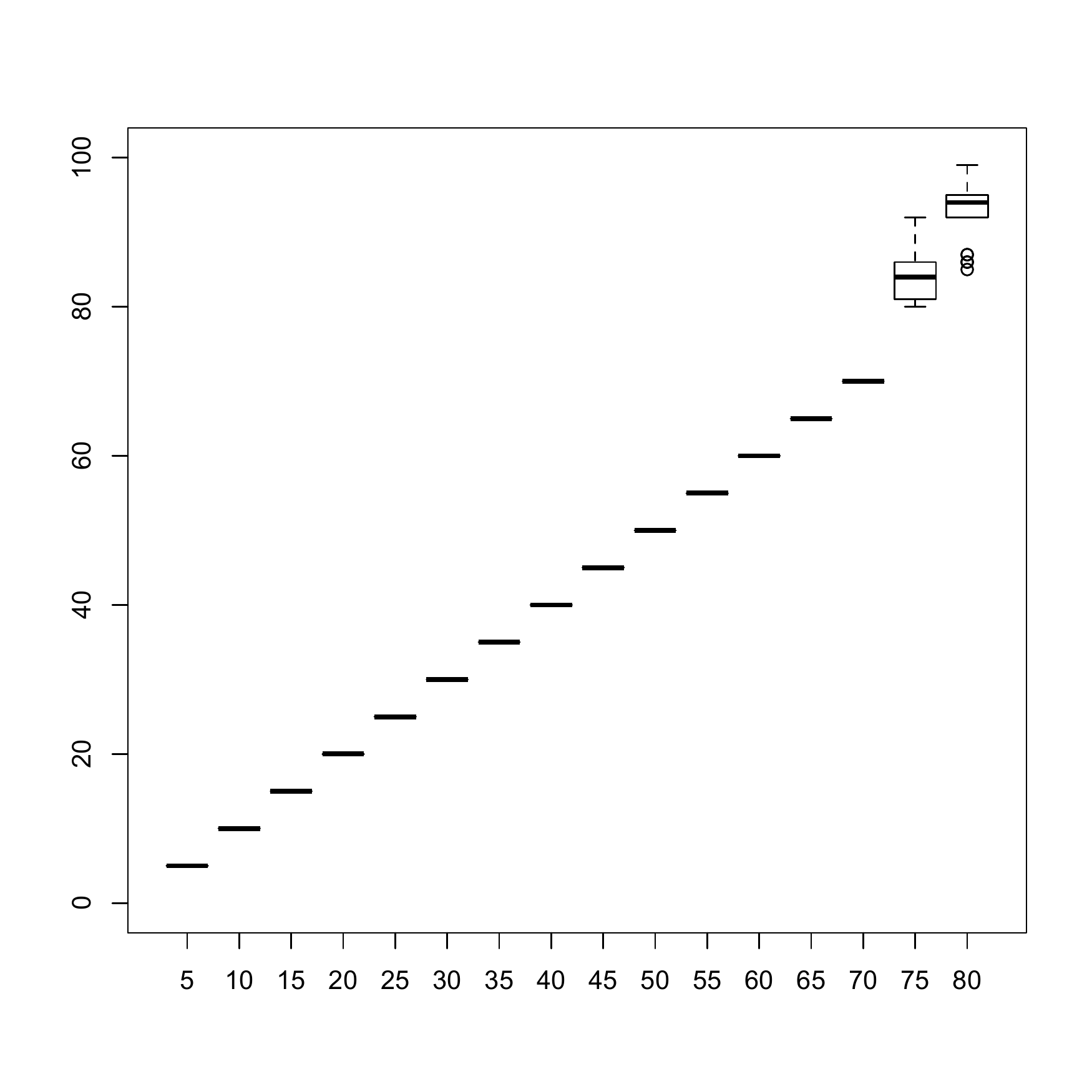}

  \caption{Boxplots of the recovery errors (first line) and recovered
    ranks (second line) using NNM (left) and WSST (right) of a
    $500\times500$ rank $r$ matrix with $r$ between $5$ and $80$
    (x-axis)}
  \label{fig:phase-transision}
\end{figure}

\subsubsection{Image inpainting}

In Figure~\ref{fig:inpainting-examples}, we consider the
reconstruction of four test images (``lenna'', ``fingerprint'',
``flinstones'' and ``boat''). Each test image has $512 \times 512$
pixels, and is of rank $50$. We only observe $30\%$ of the pixels,
picked uniformly at random, with no noise. The observations are given
in the first line of Figure~\ref{fig:inpainting-examples}, where
non-observed pixels are represented by white. The second line gives
the reconstruction obtained using NNM. The third line shows the
difference between the true image and the recovery by NNM, where blue
is perfect recovery and red is bad recovery. The fourth line shows the
reconstruction using WSST and the fifth shows the difference between
the true image and recovery by WSST.

On all four images, the recovery is much better using WSST, in
particular on the fingerprint and flinstones images. This can be
understood form the fact that these two are very structured
images. The most surprising fact is that all the four reconstructions
using NNM have rank 150 (because of the way we choose $\lambda$, see
above), while the rank of the reconstructions obtained with WSST is
never more than 90 (with the same choice of $\lambda$). So, WSST leads
to simpler (with a lower rank, which is better in terms of
compression/description) and more accurate reconstructions. In particular, we
observe that WSST is able to recover in a more precise way the
underlying geometry of the true images (for instance, on the third
line, first column, we can recognize the shape of lenna, while this is
not the case with WSST).

\setlength{\figlength}{3cm}%
\setlength{\figwidth}{3cm}

\begin{figure}[htbp]
  \centering
  
  \includegraphics[width=\figwidth,height=\figlength]{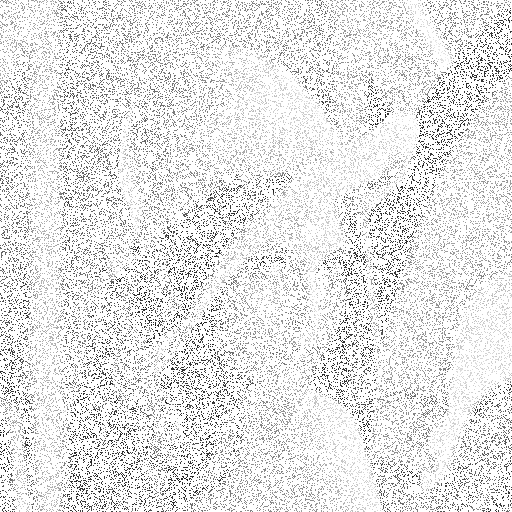} 
  \includegraphics[width=\figwidth,height=\figlength]{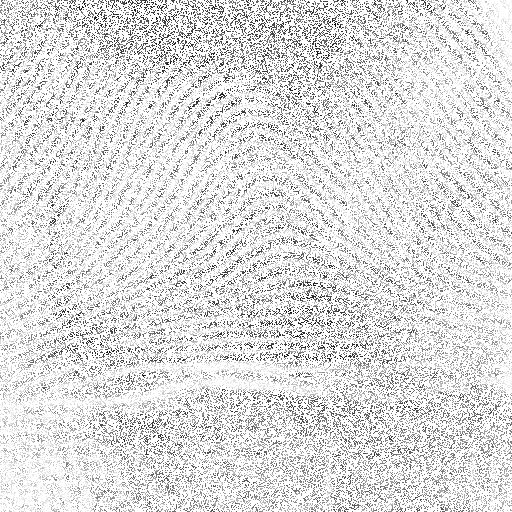}
  \includegraphics[width=\figwidth,height=\figlength]{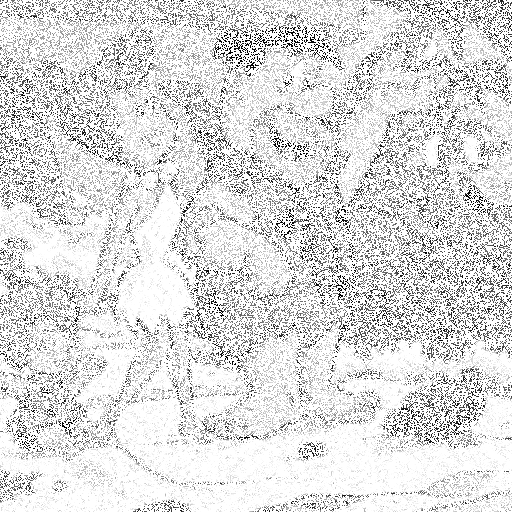}
  \includegraphics[width=\figwidth,height=\figlength]{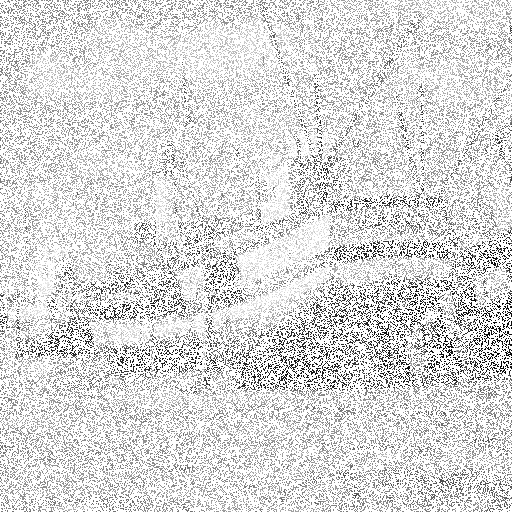}
  \\
  \includegraphics[width=\figwidth,height=\figlength]{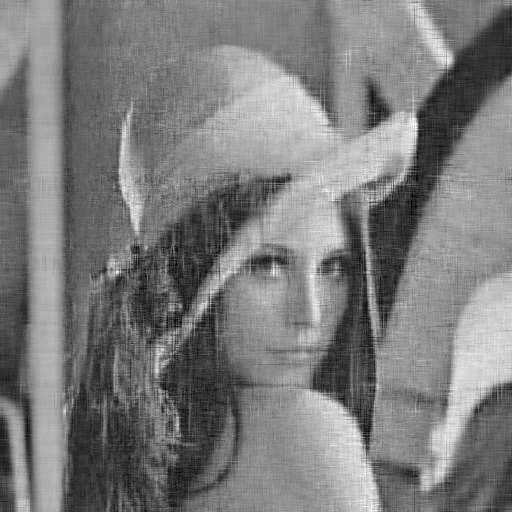}
  \includegraphics[width=\figwidth,height=\figlength]{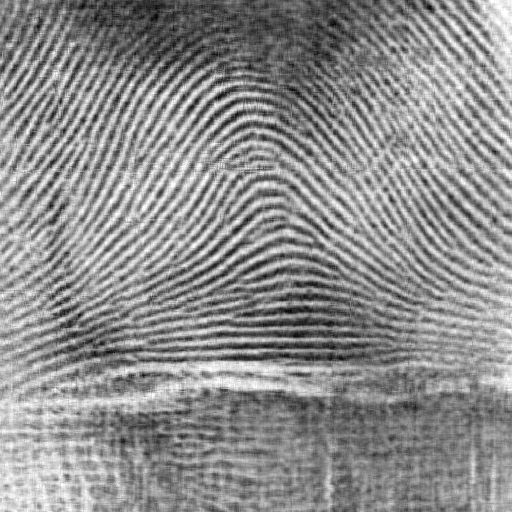} 
  \includegraphics[width=\figwidth,height=\figlength]{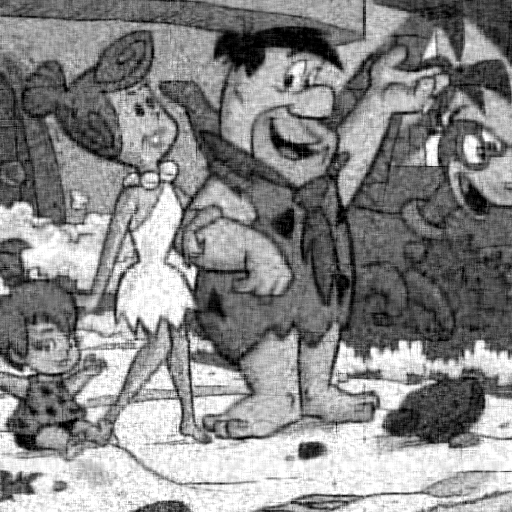}
  \includegraphics[width=\figwidth,height=\figlength]{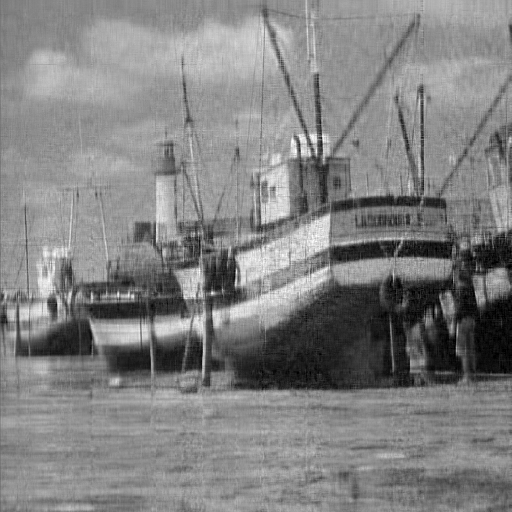}
  \\ 
  \includegraphics[width=\figwidth,height=\figlength]{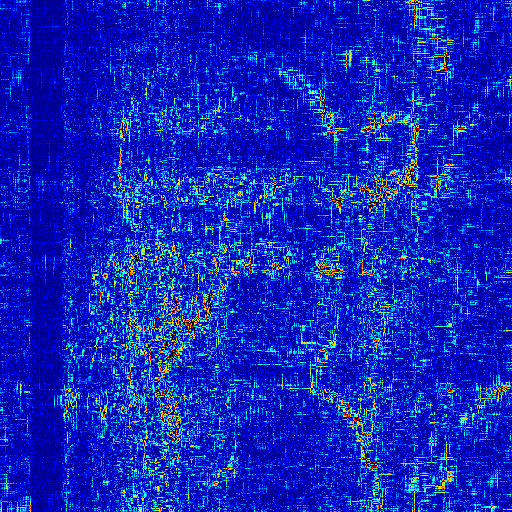}
  \includegraphics[width=\figwidth,height=\figlength]{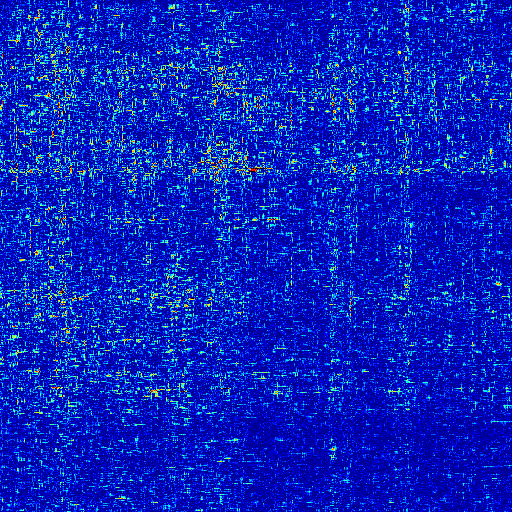}
  \includegraphics[width=\figwidth,height=\figlength]{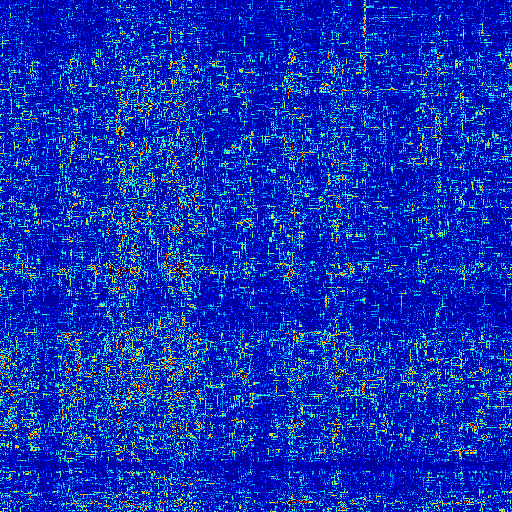}
  \includegraphics[width=\figwidth,height=\figlength]{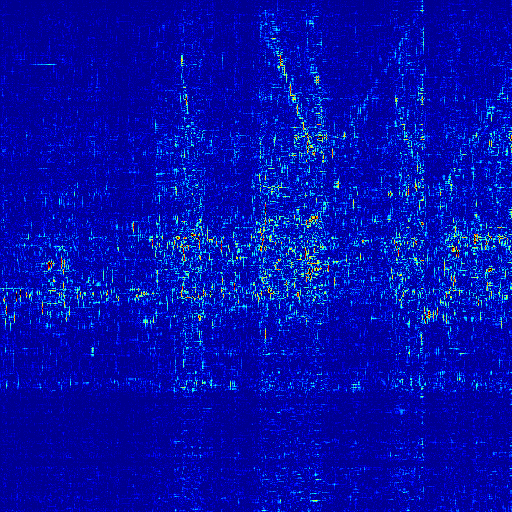}
  \\  
  \includegraphics[width=\figwidth,height=\figlength]{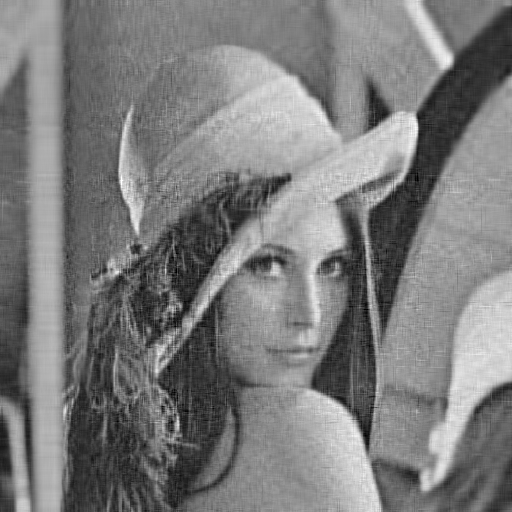} 
  \includegraphics[width=\figwidth,height=\figlength]{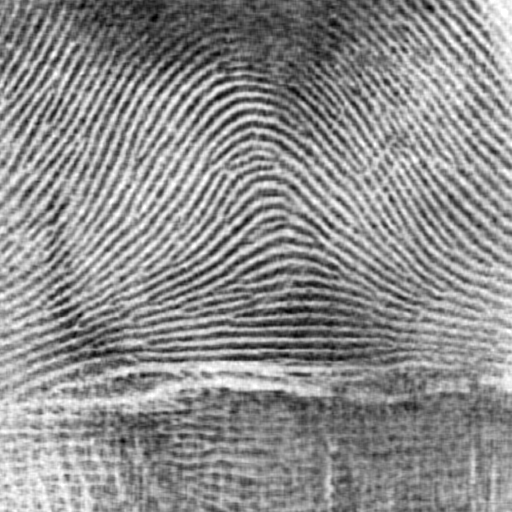}
  \includegraphics[width=\figwidth,height=\figlength]{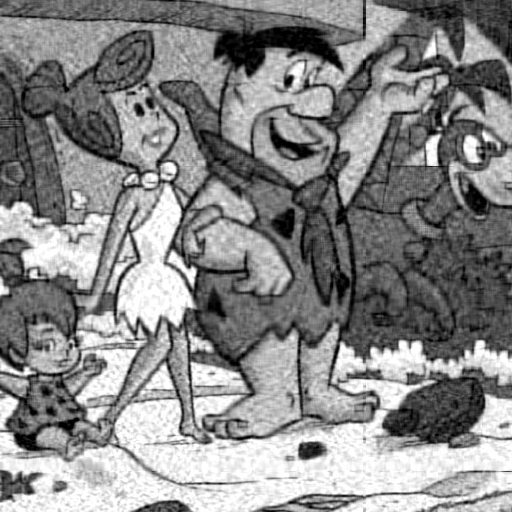}
  \includegraphics[width=\figwidth,height=\figlength]{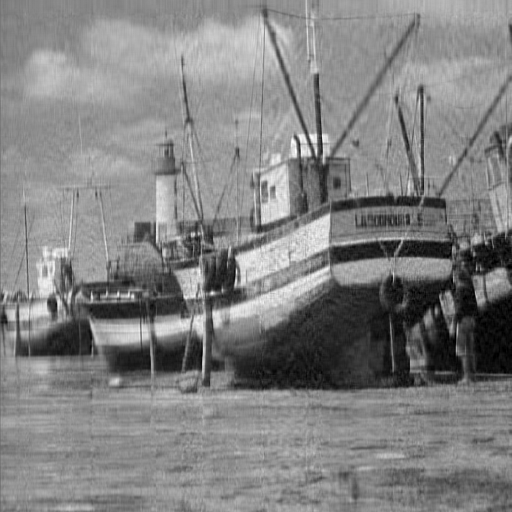}
  \\
  \includegraphics[width=\figwidth,height=\figlength]{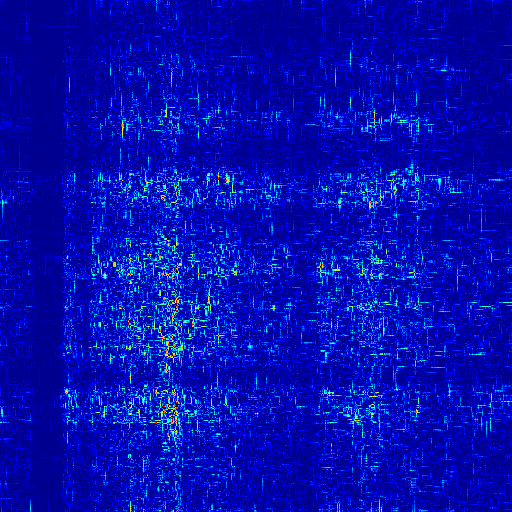} 
  \includegraphics[width=\figwidth,height=\figlength]{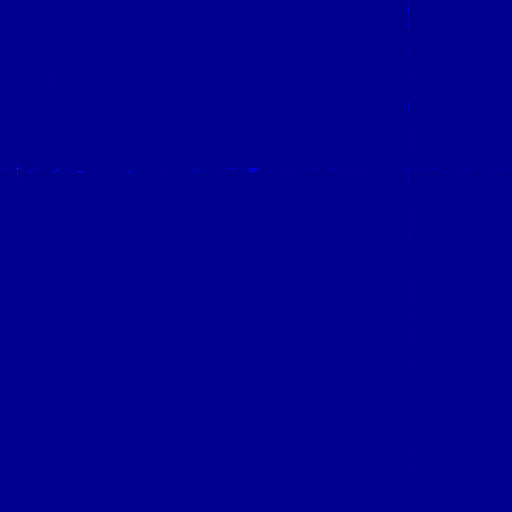}
  \includegraphics[width=\figwidth,height=\figlength]{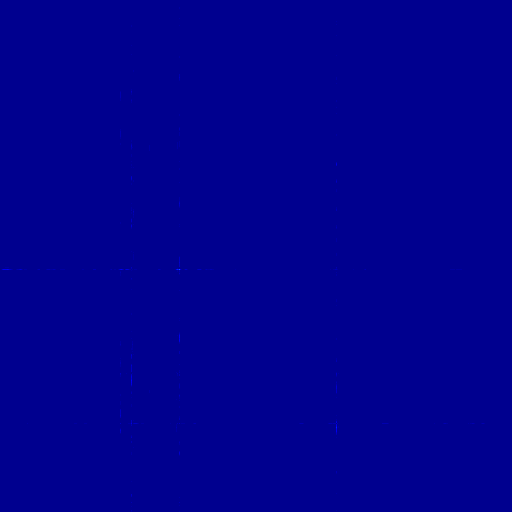}
  \includegraphics[width=\figwidth,height=\figlength]{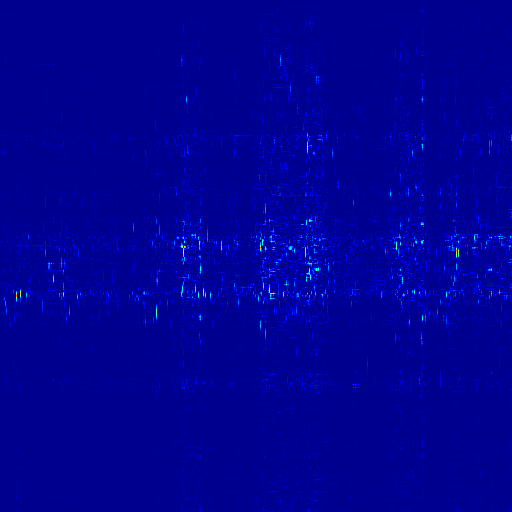}
  \\

  \caption{Image reconstruction using NNM and WSST. \emph{First line}:
    observed pixels (white means non-observed). \emph{Second line}:
    reconstruction using NNM. \emph{Third line}: difference between
    truth and NNM (red is bad, blue is good). \emph{Fourth line}:
    recovery using WSST. \emph{Fifth line}: difference between truth
    and WSST.}
  \label{fig:inpainting-examples}
\end{figure}

\subsubsection{Collaborative filtering}

Now, we consider matrix completion for a real dataset: the MovieLens
data. It contains 3 datasets, available on
\texttt{http://www.grouplens.org/}:
\begin{itemize}
\item \texttt{movie-100K:} 100,000 ratings for 1682 movies by 943 users
\item \texttt{movie-1M:} 1 million ratings for 3900 movies by 6040 users
\item \texttt{movie-10M:} 10 million ratings and 100,000 tags for
  10681 movies by 71567 users
\end{itemize}
The ranks of the users are integers between $1$ and $5$. In each 3
datasets, each user has rated at least $20$ movies. For our
experiments, we simply choose uniformly at random half of the ratings
of each user to form a subset $\Gamma$ of the entire subset $\Omega$
or ratings. Then, based on the ratings in $\Gamma$, we try to predict
the ratings in $\Omega - \Gamma$. Since many entries are missing, we
measure the accuracy of completion by computing the relative error in
$\Omega - \Gamma$. If $\hat A$ is a reconstruction matrix, we
reproduce in Table~\ref{tab:collaborative} below the values of
\begin{equation}
  \mathrm{err} = \norm{\cP_{\Omega-\Gamma}(\hat A) -
    \cP_{\Omega-\Gamma}(A_0)}_2 / \norm{\cP_{\Omega-\Gamma}(A_0)}_2,
\end{equation}
together with the rank used for the reconstruction. We observe in
Table~\ref{tab:collaborative} that WSST improves a lot upon NNM on
each datasets. The most surprising fact is that the rank used by WSST
is much smaller than the one used by NNM, while leading at the same
time to strong prediction improvements. For \texttt{movie-1M} for
instance, the prediction error of WSST is $30\%$ better than NNM,
while NNM solution has rank 200 and the WSST has rank 40. Once again,
we can conclude on this example that WSST gives both much simpler
reconstructions, and better prediction accuracy. Note that we
considered a maximum rank equal to 200 for the \texttt{movie-100K} and
\texttt{movie-1M} datasets, and equal to 50 for \texttt{movie-10M} (to
make this problem computationally tractable on a normal computer).

\begin{table}[htbp]
  \centering
  \footnotesize
  \begin{tabular}{lcccccc}
    \multicolumn{7}{l}{\rule{12.5cm}{1pt}} \\
    &  &  & \multicolumn{2}{c}{relative error} &
    \multicolumn{2}{c}{rank} \\
    & $n_1 / n_2$ & $m$ & \multicolumn{2}{l}{\rule{2.9cm}{0.5pt}} &
    \multicolumn{2}{l}{\rule{2.6cm}{0.5pt}} \\    
    & & & NNM & WSST & NNM & WSST\\    
    \texttt{movie-100K:} & 943/1682 & 1.00e+5 & 3.92e-01 
    & 3.30e-01 & 128 & 33 \\
    \texttt{movie-1M:} & 6040/3702 & 1.00e+6 & 3.83e-01 & 2.70e-01
    & 200 & 40 \\
    \texttt{movie-10M:}  & 71567/10674 & 9.91e+6 & 2.76e-01
    & 2.36e-01 & 50 & 5 \\
    \multicolumn{7}{l}{\rule{12.5cm}{1pt}}
  \end{tabular}
  \caption{Relative reconstruction errors for the MovieLens datasets.}
  \label{tab:collaborative}
\end{table}

\section{Proofs}
\label{sec:proofs}

\subsection{Proofs for Section~\ref{sec:cs-vectors}}
\label{sec:proofs-cs}

We denote by $\ell_p^M$ the space $\R^M$ endowed with the $\ell_p$
norm. The unit ball there is denoted by $B_p^M$. We also denote the
unit Euclidean sphere in $\R^M$ by $\cS^{M-1}$. We denote by
$(e_1,\ldots,e_N)$ the canonical basis of $\R^N$ and for any
$I\subset\{1,\ldots,N\}$ denote by $\R^I$ the subspace of $\R^N$
spanned by $(e_i:i\in I)$.  Let $A = [A_{\{1\}}, \ldots, A_{\{N\}}]$
be a matrix from $\R^N$ to $\R^m$, where $A_{\{i\}}$ denotes the
$i$-th column vector of $A$. Let $x \in \R^N$ and $I$ an arbitrary
subset of $\{1,\ldots,N\}$. We define $A_I = [A_{\{i\}} : i \in I]$
the matrix from $\R^{I}$ to $\R^m$ with columns vectors $A_{\{i\}}$
for $i\in I$. We denote by $x_I$ the vector in $\R^{I}$ with
coordinates $x_i$ for $i\in I$, where $x_i$ is the $i$-th coordinate
of $x$. We denote by $x^I$ the vector of $\R^N$ such that $x_i^I=0$
when $i\notin I$ and $x_i^I=x_i$ when $i\in I$.  If $w\in\R^N$ has non
negative coordinates, we denote by $wx$ the vector $(w_1
x_1,\ldots,w_N x_N)$ and by $x/w$ the vector
$(x_1/w_1,\ldots,x_N/w_N)$ with the previous convention in case where
$w_i=0$ for some $i$. We denote by $|x|$ the vector
$(|x_1|,\ldots,|x_N|)$. The support of $x$ is denoted by $I_x$, this
is the set of all $i\in\{1,\ldots,N\}$ such that $x_i\neq0$. We also
consider the $w$-weighted $\ell_1^N$-norm
 \begin{equation}
   \label{eq:l1weightednorm}
   |x|_{1,w} = \sum_{i=1}^N \frac{|x_i|}{w_i}.
\end{equation}
Note that $|\cdot|_{1,w}$ is a norm only when restricted to
$\R^{I_w}$, where $I_w$ is the support of $w$.

We start with the well-known null space property and dual
characterization~\cite{MR2236170} of exact reconstruction of a vector
by $\ell_1$-based algorithms.

\begin{proposition}
  \label{prop:equivalence-reconstruction-exacte}
  Let $x, w \in \R^N$ and denote by $I_x$ \textup(resp. $I_w$\textup)
  the support of $x$ \textup(resp. $w$\textup). The following points
  are equivalent\textup:
  \begin{enumerate}
  \item $\Delta_w(Ax) = x$,
  \item $I_x \subset I_w$ and for any $h \in \ker A_{I_w}$ such that
    $h \neq 0$ then
    \begin{equation*}
      \Norm{\Big(\frac{h}{w_{I_w}}\Big)_{I_x^\complement}}_1 +
      \Inr{\sgn(x_{I_x}), \Big( \frac{h}{w_{I_w}}\Big)_{I_x}} > 0,
    \end{equation*}
  \item $I_x\subset I_w$ and there exists $Y \in (\ker
    A_{I_w})^{\bot}$ such that $(w_{I_w}Y)_{I_x}={\rm sign}(x_{I_x})$
    and $|(w_{I_w}Y)_{I_x^\complement}|_\infty < 1$.
  \end{enumerate}
\end{proposition}

\begin{proof}
  It follows from~\eqref{eq:values-weight-algo} that, under each one
  of the three conditions, we have $I_x\subset I_w$. Therefore, to
  simply notations, we can work as if the ambient space were
  $\R^{I_w}$. Hence, without loss of generality, we assume that
  $\R^{I_w}=\R^N$. We also denote by $I=I_x$ the support of $x$.
  
  [Point~2. entails Point~1.] Using standard arguments (see for
  instance~\cite{MR0274683}), we can see that the subgradient of
  $|\cdot|_{1,w}$ at $x \in \R^N$ is the set
  \begin{equation}
    \label{eq:subgradient-weighted-norm}
    \begin{split}
      \partial |x|_{1,w} = \big \{t \in \R^N : t_i &=  \sgn(x_i)/w_i
      \text{ when } x_i\neq0 \\
      &\text{ and } |t_i|\leq 1/w_i \text{ when }
      x_i=0\big\}.
    \end{split}  
  \end{equation}
  Using the definition of the subgradient of $|\cdot|_{1,w}$ at $x$,
  it follows that for any $h \in \R^N$,
  \begin{equation*}
    |x+h|_{1,w} \geq |x|_{1,w} + |(h/w)_{I^\complement}|_1 +
    \inr{\sgn(x_I),(h/w)_I}.
  \end{equation*}
  Thus, if Point~2 holds then for any $h \in \ker A$ such that $h \neq
  0$,
  \begin{equation*}
    |x+h|_{1,w} > |x|_{1,w}
  \end{equation*}
  and thus Point~1 is satisfied.
  
  [Point~3. entails Point~2.] Let $Y \in (\ker A)^{\bot}$ such that
  $(wY)_I=\sgn(x_I)$ and $|(wY)_{I^\complement}|_\infty<1$.  For any $h\neq 0$
  in $\ker A$, we have
  \begin{align*}
    |(h/w)_{I^\complement}|_1 & + \inr{\sgn(x_I),(h/w)_I} =
    \inr{\sgn(x)^I+\sgn(h)^{I^\complement}, h/w} \\
    &= \inr{(\sgn(x)/w)^I + (\sgn(h)/w)^{I^\complement}, h} \\
    &= \inr{(\sgn(x)/w)^I + (\sgn(h)/w)^{I^\complement} - Y,h} \\
    &= \inr{(\sgn(h)/w)_{I^\complement}-Y_{I^\complement},h_{I^\complement}}
    = \sum_{i\in I^\complement}\frac{h_i}{w_i}
    \big(\sgn(h_i)-w_iY_i\big)>0,
  \end{align*}
  where we used Point~3 in the fourth inequality.
  
  [Point~1. entails Point~3.]  This follows from classical results on
  the minimization of a convex function over a convex set
  (cf.~\cite{MR0274683}). Nevertheless, we provide a direct proof
  following the argument of \cite{MR2236170}. Denote by $\{ e_1,
  \ldots, e_N \}$ the canonical basis in $\R^N$ and by $B_{1,w}^N$ the
  unit ball associated to the $w$-weighted $\ell_1^N$-norm:
  \begin{equation}
    \label{eq:weighted-l1-norm}
    B^N_{1,w} = \{t \in\R^N : |t|_{1,w} \leq 1 \}.
  \end{equation}
  If $x$ is the unique solution of~\eqref{eq:general-weighted-algo}
  then $|x|_{1,w} B_{1,w}^N \cap (x + \ker A ) = \{ x \}$. Then by a
  duality argument (for instance Hahn-Banach Theorem for the
  separation of convex sets), there exists $Y \in \R^N$ such that $x +
  \ker A \subset \Gamma_1$, where $\Gamma_1 = \{t : \inr{t,Y} = 1 \}$
  and $|x|_{1,w} B_{1,w}^N \subset \Gamma_{\leq 1}$, where
  $\Gamma_{\leq 1} =\{t : \inr{t,Y} \leq 1 \}$. Introduce $F_{1,w}(x)
  = |x|_{1,w} \conv(w_ie_i : x_i\neq 0)$, the face of $|x|_{1,w}
  B_{1,w}^N$ containing $x$. By moving the hyperplan $\Gamma_1$, we
  can assume that $|x|_{1,w} B_{1,w}^N \cap \Gamma_1 \subset
  F_{1,w}(x)$. Since $|x|_{1,w} B_{1,w}^N\subset
  \Gamma_{\leq1}$, we have $\sup_{t\in |x|_{1,w}
    B_{1,w}^N}\inr{t,Y}\leq 1$ thus
  $|(wY)|_\infty\leq1/|x|_{1,w}$. Moreover, $x \in \Gamma_1$ so $1 =
  \inr{x,Y} \leq |x|_{1,w} |(wY)|_\infty \leq 1$ because
  $|(wY)|_\infty \leq 1 / |x|_{1,w}$. This is the equality case in
  H{\"o}lder's inequality, so it follows that $(wY)_I = \sgn(x_I) /
  |x|_{1,w}$. Then, for any $i \notin I$, $|x|_{1,w}w_ie_i\in
  |x|_{1,w} B^N_{1,w}$, thus $\inr{|x|_{1,w}w_ie_i,Y} \leq 1$ and
  $|x|_{1,w}w_ie_i\notin F_{1,w}(x)$, so $|x|_{1,w} w_i e_i\notin
  \Gamma_1$ thus $\inr{|x|_{1,w}w_ie_i,Y}< 1$. That is,
  $|(wY)_{I^\complement}|_\infty < 1/|x|_{1,w}$. Finally, for any
  $h\in\ker A$, $1=\inr{x+h,Y}=\inr{x,Y}+\inr{h,Y}=1+\inr{h,Y}$, thus
  $\inr{h,Y} = 0$ and $Y \in (\ker A)^\bot$. Then, we normalize $Y$ by
  $|x|_{1,w}$ to obtain Point~3.
\end{proof}

Both Criterions~2 and~3 in
Proposition~\ref{prop:equivalence-reconstruction-exacte} can be used
to characterize the exact reconstruction of a vector $x$ by the
$\ell_1$-weighted algorithm. The vector $Y$ of Criterion~3 is now
called an \emph{exact dual certificate}
(cf. \cite{MR2236170,Gross}). We will use Criterion~3 and the
construction of an exact dual certificate from~\cite{MR2236170} to
prove Theorems~\ref{thm:A} and~\ref{thm:B}. Note that Criterion~2
together with the construction of an \emph{inexact dual certificate}
(cf.~\cite{Gross}) can also be used. Nevertheless, we do not present
this construction here since it does not improve the statement of
Theorem~\ref{thm:B}.

\subsubsection{Proof of Theorem~\ref{thm:A}}
\label{sec:TheoA}
  In the same way as we did in the proof of
  Proposition~\ref{prop:equivalence-reconstruction-exacte}, we can
  work as if the ambient space were $\R^{I_w}$ and assume, without
  loss of generality, that $\R^{I_w}=\R^N$. We denote by $I$ the
  support of $x$. We prove first that when $\Delta_1(Ax) = x$, then
  $A_I$ is injective. Indeed, suppose that there exists some
  $h\in\R^{I}$ such that $h\neq0$ and $A_Ih=0$. Denote by $h^0 \in
  \R^N$ the vector such that $h^0_I = h$ and
  $h^0_{I^\complement}=0$. We have $h^0 \neq 0$ and $A h^0 = A_I h_I =
  0$. In particular, for any $\lambda \neq 0$, $\lambda h \in \ker A -
  \{0\}$. Therefore, since $x$ is the unique solution of the Basis
  Pursuit algorithm, it follows from Point~2 of
  Proposition~\ref{prop:equivalence-reconstruction-exacte} (applied to
  the weight vector $w=(1,\ldots,1)$), that, for every $\lambda \neq
  0$, $\inr{\sgn(x_I), \lambda h^0_I} > 0.$ This is not possible, so
  $A_I$ is injective.
  
  Since $\Delta_1(Ax)=x$, the decoder $\Delta_2$ is given here by
  \begin{equation*}
    \Delta_2(Ax) \in \argmin_{t \in \R^N} \Big(\sum_{i=1}^N
    \frac{|t_i|}{|x_i|} : At = Ax \Big).
  \end{equation*}
  Therefore, according to~\eqref{eq:values-weight-algo}, we have
  $\Delta_2(Ax)_i = 0$ for any $i \notin I$, that is ${\rm
    supp}(\Delta_2(Ax))\subset I$. As a consequence
  $A_Ix_I=Ax=A\Delta_2(Ax)=A_I\Delta_2(Ax)_I$ and $A_I$ is injective
  thus, $x_I=\Delta_2(Ax)_I$. Since $x_{I^\complement} = 0 =
  \Delta_2(Ax)_{I^\complement}$, we have $x = \Delta_2(Ax)$.

\subsubsection{Proof of Theorem~\ref{thm:B}}
\label{sec:TheoB}

We adapt to our setup the ``dual certificate'' introduced
in~\cite{MR2236170} and consider
\begin{equation}
  \label{eq:exact-dual-certificate}
  Y^0 = A^\top A_I (A_I^\top A_I)^{-1} \Big( \frac{\sgn(x)}{w} \Big)_I.
\end{equation}
In particular, we have $Y^0 \in \im (A^\top) = (\ker A)^\bot$ and
\begin{equation*}
  Y_I^0 = A^\top_I A_I(A_I^\top A_I)^{-1}
  \Big(\frac{\sgn(x)}{w}\Big)_I = \Big(\frac{\sgn(x)}{w}\Big)_I.
\end{equation*}
Thus, we have $(wY^0)_I={\rm sgn}(x_I)$. In view of
Proposition~\ref{prop:equivalence-reconstruction-exacte}, it only
remains to prove that $|(wY^0)_{I^\complement}|<1$ with high
probability.  For $0 < \delta < 1$ and $\mu > 0$, we consider the
events
 \begin{equation}
   \label{eq:Omega0}
   \Omega_0(I,\delta) = \big\{(1-\delta) |y|_2^2\leq |A_I y|_2^2 \leq
   (1 + \delta) |y|_2^2, \quad \forall y \in \R^{I} \big\}
 \end{equation}
 and
 \begin{equation}
   \label{eq:Omega1}
   \Omega_1(I,\mu) = \big\{\max_{i\in I^\complement} |A_I^\top  A_{\{i\}}|_2 <
   \mu \big\}.
 \end{equation}
 First, note that since $A_I^\top A_I - Id$ is Hermitian, we have
 \begin{equation*}
   \norm{A_I^\top A_I - Id}_{2 \rightarrow 2} = \sup_{|y|_2=1}
   \big| |A_I y|_2^2 - 1 \big|.
 \end{equation*}
 Thus, on $\Omega_0(I,\delta)$, we have $ \norm{A_I^\top A_I
   -Id}_{2\rightarrow2}\leq \delta$ and so for any $y\in\R^{I}$,
 $\big| (A_I^\top A_I)^{-1} y \big|_2 \leq (1 - \delta)^{-1}
 |y|_2$.  In
 particular,
 \begin{equation*}
   \label{eq:Intermediare-0}
   \Big|(A_I^\top A_I)^{-1} \Big(\frac{\sgn(x)}{w}\Big)_I \Big|_2 \leq
   \frac{1}{1 - \delta} \Big| \Big(\frac{\sgn(x)}{w}\Big)_I \Big|_2 =
   \frac{1}{1 - \delta} \big| \big( 1 / w \big)_I \big|_2.
 \end{equation*}
 Then, it follows that, on $\Omega_0(I, \delta) \cap \Omega_1(I, x)$
 and under condition $(A0)(I, (1-\delta) / \mu)$,
 \begin{align*}
   |(wY^0)_{I^\complement}|_{\infty} &= \max_{i\in I^\complement}
   \Big| w_i A_{\{i\}}^\top
   A_I (A_I^\top A_I)^{-1} \Big(\frac{\sgn(x)}{w}\Big)_I \Big| \\
   &\leq \max_{i\in I^\complement} w_i \max_{i\in I^\complement}
   \Big|\inr{A_I^\top
     A_{\{i\}}, (A_I^\top A_I)^{-1} \Big(\frac{\sgn(x)}{w}\Big)_I}\Big| \\
   &\leq \max_{i\in I^\complement} w_i \max_{i\in I^\complement}
   \big|A_I^\top
   A_{\{i\}}\big|_2 \Big|(A_I^\top A_I)^{-1}
   \Big(\frac{\sgn(x)}{w}\Big)_I\Big|_2 \\
   &< \frac{\mu}{1-\delta} \max_{i\in I^\complement} w_i
   \big|\big(1/w\big)_I\big|_2 \leq 1.
 \end{align*}
 Then, Theorem~\ref{thm:B} follows from the probability estimates of
 $\Omega_0(I,\delta) \cap \Omega_1(I,\mu)$ provided in the next lemma.

 \begin{Lemma}
   \label{lem:proba-estimates}
   Let $A = m^{-1/2}\big(g_{i, j}\big)$ be a $m\times N$ matrix where the
   $g_{i, j}$'s are i.i.d. standard Gaussian variables. Assume that
   \begin{equation*}
     m \geq c_0\max\Big[ \frac{s}{\delta^2}, \frac{s\log N}{ \mu^2} \Big].
   \end{equation*}
   With probability larger than $1 - 2\exp(-c_1 m \delta^2) -
   \exp(- c_2\mu^2 m /s)$, we have
   \begin{equation*}
     (1 - \delta) |y|_2^2 \leq |A_I y|_2^2 \leq (1 + \delta) |y|_2^2,
     \quad \forall y \in \R^I
  \end{equation*}
  and $\max_{i\in I^\complement} |A_I^\top A_{\{i\}}|_2 < \mu.$
\end{Lemma}

\begin{proof}
  For the sake of completeness, we recall here the classical
  $\eps$-net argument to prove the first statement of
  Lemma~\ref{lem:proba-estimates}. It is enough to prove that
  $\sup_{y\in \cS^I} | |A_Iy|_2^2 - 1 | \leq \delta$, where $\cS^I$ is
  the set of unit vectors of $\ell_2^N$ supported on $I$. First, note
  that
  \begin{equation*}
    \sup_{y \in \cS^I} \big| |A_Iy|_2^2 - 1 \big| = \sup_{y \in \cS^I}
    | \inr{Ty, y} | = \norm{T}_{2\rightarrow2},
  \end{equation*}
  where $T : \R^I \rightarrow \R^I$ is the symmetric operator $A^\top
  A - I_d$. Let $\Lambda \subset \cS^I$ be a $1/4$-net of $ \cS^I$ for
  the $\ell_2$ metric with a cardinality smaller than $9^s$ (the
  existence of such a net follows from a volumetric argument,
  see~\cite{MR1036275}). For any $y\in\cS^I$, there exists $z \in
  \Lambda$ such that $y = z + u$ with $\Norm{u}_2 \leq
  1/4$ and  therefore,
  \begin{equation*}
    |\inr{Ty,y}|\leq |\inr{Tz,z}|+|\inr{Tu,u}|+2|\inr{Tz,u}|\leq
    \max_{z\in\Lambda}|\inr{Tz, z}|+\frac{9\norm{T}_{2\rightarrow2}}{16}.
  \end{equation*}
   Hence, $\norm{T}_{2 \rightarrow
    2}\leq (16/7) \max_{z \in \Lambda}|\inr{Tz, z}|$, and it is enough
  to control the supremum of $y \rightarrow |\inr{Ty,y}|$ over
  $\Lambda$ instead of $\cS^I$.
  
  Let $y\in\Lambda$. We denote by $G_1 / \sqrt{m}, \ldots, G_m /
  \sqrt{m}$ the row vectors of $A$ where $G_1, \ldots, G_m$ are $m$
  independent standard Gaussian vectors of $\R^N$. We have $\inr{Ty,
    y} = m^{-1}\sum_{i=1}^m \inr{G_i,y}^2 - 1$. Since $\| \inr{G,y}^2
  \|_{\psi_1} = \| \inr{G,y} \|_{\psi_2}^2$, it follows from Bernstein
  inequality for $\psi_1$ random variables~\cite{vanderVaartWellner}
  that
  \begin{equation*}
    \P\big[ |\inr{Ty,y}| \leq \delta\big] \geq 1 - 2\exp(-c_1 m\delta^2),
  \end{equation*}
  and a union bound yields 
  \begin{equation*}
    \P\big[ |\inr{Ty,y}| \leq \delta \; , \; \forall y \in \Lambda\big] \geq
    1 - 2\exp( s \log 9 - c_1m \delta^2). 
  \end{equation*}
  Combining the $\eps$-net argument with this probability estimate we
  obtain that when $m \geq c_2 s/\delta^2$ then
  $\norm{T}_{2\rightarrow2}\leq \delta$ with probability at least $1 -
  2\exp\big( -c_3m\delta^2\big)$.
  
  Now, we turn to the second part of the statement. Let $i\in
  I^\complement$. The $i$-th column vector of $A$ is
  $A_{\{i\}}=G_i/\sqrt{m}=(g_{i1},\ldots,g_{im})^\top/\sqrt{m}$ where
  the $G_i$'s are independent standard Gaussian vectors of $\R^m$. Let
  $q \geq 2$ to be chosen later. By Markov inequality,
  \begin{equation}
    \label{eq:markov-theoB}
    \Pro\Big[\Big|A_I^\top A_{\{i\}} \Big|_2\geq
    \mu\Big]=\Pro\Big[\Big|\sum_{j=1}^m g_{ij}G_{jI}\Big|_2\geq
    m\mu\big]\leq (m\mu)^{-q}\E\Big|\sum_{j=1}^mg_{ij}G_{jI}\Big|_2^q.
  \end{equation}
  Now, we use the vectorial version of Khintchine inequality
  conditionally to $G_{1J},\ldots,G_{mJ}$, to obtain, for some
  absolute constant $c_4$,
  \begin{equation*}
    \Big(\E_g\Big|\sum_{j=1}^mg_{ij}G_{jI}\Big|_2^q\Big)^{1/q}\leq
    c_4\sqrt{q}\Big(\E_g\Big|\sum_{j=1}^mg_{ij}G_{jI}\Big|_2^2\Big)^{1/2}=c_4
    \sqrt{q}\Big(\sum_{j=1}^m\big| G_{jI}\big|_2^2\Big)^{1/2}.
  \end{equation*}
  It follows that
  \begin{equation*}
    \E\Big|\sum_{j=1}^mg_{ij}G_{jI}\Big|_2^q\leq \big(c_4^2 q m s\big)^{q/2}.
  \end{equation*}
  Hence, in (\ref{eq:markov-theoB}) for
  $q=\big(\mu/(2c_4^2)\big)^2(m/s)$, we obtain
  \begin{equation*}
    \Pro\Big[\Big|A_I^\top A_{\{i\}} \Big|_2\geq \mu\Big]\leq
    \exp\Big(-\frac{\mu^2 m \log 2}{s(2c_4^2)^2}\Big).
  \end{equation*}
  The result follows now from an union bound.
\end{proof}

\subsubsection{Proof of Theorem~\ref{thm:C}}
\label{sec:TheoC}

\begin{proof}
  Assume that $\Delta_r(Ax) = x$ and define $y = \Delta_{r+1}(Ax)$. By
  construction of $y$, we have $\supp(y) \subset \supp(x)$ and $Ax =
  Ay$. So, since $A$ is injective on $\Sigma_m$ and $x-y \in
  \Sigma_m$, we have $x=y$. This proves that $\Delta_{r+1}(Ax) = x$,
  and that the sequence $(\Delta_n(Ax))_n$ is constant and equal to a
  $\lfloor m/2 \rfloor$-sparse vector starting from the $r$-th
  iteration.
  
  Now, assume that there exists an integer $r$ and $y \in
  \Sigma_{\lfloor m/2 \rfloor}$ such that $\Delta_r(Ax) =
  \Delta_{r+1}(Ax) = \cdots=y$. In particular, we have $Ay = Ax$, so
  since $A$ is injective on $\Sigma_m$ and $x-y \in \Sigma_m$, we have
  $x=y$.
\end{proof}

\subsection{Proofs for Section~\ref{cs-matrices}}
\label{sec:proofs-cs-matrices}

The next proposition shows that weighted spectral soft-thresholding
achieves the minimum of the weighted nuclear norm plus a proximity
term. Note that, however, weighted spectral soft-thresholding is not a
proximal operator, since the weighted nuclear norm is not convex. This
entails in particular that the proofs below use a direct analysis,
since we cannot use arguments based on subdifferential computations
here.
\begin{proposition}
  \label{prop:weighted-nuclear-prox}
  Let $B \in \R^{n_1 \times n_2}$, $\tau,\lambda \geq 0$ and $w_1 \geq \cdots
  \geq w_{n_1 \wedge n_2} \geq 0$. Then the minimization
  problem
  \begin{equation*}
    \min_{A \in \R^{n_1 \times n_2}} \Big\{ \frac 12 \norm{A - B}_2^2 +
    \lambda\sum_{j=1}^{n_1 \wedge n_2} \frac{\sigma_j(A)}{w_j} +
    \frac{\tau}{2} \norm{A}_2^2 \Big\}
  \end{equation*}
  has a unique solution, given by $\frac{1}{1+\tau} S_\lambda^w(B)$,
  where $S_\lambda^w(B)$ is the weighted soft-thresholding
  operator~\eqref{eq:s-lambda-w-def}.
\end{proposition}

\begin{proof}[Proof of Proposition~\ref{prop:weighted-nuclear-prox}]
  Denote for short $q = n_1\wedge n_2$ and write the SVD of $A$ as $A
  = U \Sigma V^\top = \sum_{j=1}^q \sigma_j u_j v_j^\top$ where $U =
  [u_1, \ldots, u_q]$, $V = [v_1, \ldots, v_q]$ and $\Sigma =
  \diag(\sigma_1, \ldots, \sigma_q)$. We have
  \begin{equation*}
     \norm{A - B}_2^2 =
     \norm{B}_2^2 - 2 \sum_{j=1}^{q} \sigma_j u_j^\top B v_j +
    (1+\tau) \sum_{j=1}^{q} \sigma_j^2  
  \end{equation*}
  so that we want to minimize the function
  \begin{equation*}
    \phi(U, V, \Sigma) = \frac 12 \sum_{j=1}^{q} \Big(-2\sigma_j
    u_j^\top B v_j + (1+\tau)\sigma_j^2 \Big) +
    \lambda\sum_{j=1}^{q} \frac{\sigma_j}{w_j}  
  \end{equation*}
  over $U, V, \Sigma$ with the constraints $U^\top U = I$, $V^\top V =
  I$ and $\sigma_1 \geq \ldots \geq \sigma_q \geq 0$. Using the
  variational characterization of singular values, if $B = U' \Sigma'
  V'^\top$ is the SVD of $B$, where $U' = [u_1', \ldots, u_q']$, $V' =
  [v_1', \ldots, v_q']$, $\Sigma' = \diag(\sigma_1', \ldots,
  \sigma_q')$, we know that the maximum of $u^\top B v$ over all
  vectors $u$ and $v$ subject to $|u|_2 = |v|_2 = 1$ and $u$
  orthogonal to $u_1', \ldots, u_{j-1}'$ and $v$ orthogonal to $v_1',
  \ldots, v_{j-1}'$ is achieved at $u_j'$ and $v_j'$, and is equal to
  $\sigma_j'$. So the maximum of $\phi(U, V, \Sigma)$ is achieved at
  $U = U'$ and $V = V'$, and 
  \begin{equation*}
    \phi(U', V', \Sigma) = \frac 12 \sum_{j=1}^{q}  \Big( -2  \sigma_j
    \sigma_j' + (1+\tau) \sigma_j^2  + 2\lambda\frac{\sigma_j}{w_j}\Big).
  \end{equation*}
  It is easy to see that for each $j$ the the minimum over $\sigma_j$
  is achieved at $\sigma_j = \frac{1}{1+\tau} (\sigma_j' -
  \frac{\lambda}{w_j})_+$, which is non-increasing.
\end{proof}

As mentioned before, $S_\lambda^w$ is not a proximal operator. A nice
property about proximal operators is that they are firmly
non-expansive, see \cite{MR0274683}. Namely, if $T$ is the proximal
operator of some convex function over an Hilbert space $H$, then we
have
\begin{equation*}
  \norm{T x - Ty}^2 \leq \norm{x - y}^2 - \norm{x - y - (T x - T
    y)}^2
\end{equation*}
for any $x, y \in H$. However, it turns out that we can prove, using a
direct analysis, that $S_\lambda^w$ is non-expansive. Once again, the
proof uses a direct and technical analysis (since we cannot use
arguments based on subdifferential computations), while the property of
firm-nonexpansivity of proximal operators is an easy consequence of
their definition.

\begin{proposition}
  \label{prop:S-w-lipshitz}
  Let $w_1 \geq \cdots \geq w_{n_1 \wedge n_2} \geq 0, \lambda \geq
  0$. Then, for any $A, B \in \R^{n_1 \times n_2}$, we have
  \begin{equation*}
    \norm{S_\lambda^w(A) - S_\lambda^w(B)}_2 \leq \norm{A - B}_2.
  \end{equation*}
\end{proposition}

\begin{proof}[Proof of Proposition~\ref{prop:S-w-lipshitz}]
  Let us assume without loss of generality that $\lambda = 1$. Write
  the SVD of $A$ and $B$ as $A = U_1 \Sigma_1 V_1^\top$ and $B = U_2
  \Sigma_2 V_2^\top$ where
  $\Sigma_1=\diag[\sigma_{1,1},\ldots,\sigma_{1,r_1}]$,
  $\Sigma_2=\diag[\sigma_{2,1},\ldots,\sigma_{2,r_2}]$ and $r_1$
  (resp. $r_2$) stands for the rank of $A$ (resp. $B$).  We also write
  for short $\bar A = S_1^w(A) = U_1 \bar \Sigma_1 V_1^\top$ and $\bar
  B = S_1^w(B) = U_2 \bar \Sigma_2 V_2^\top$ where $\bar \Sigma_1 =
  \diag[ (\sigma_{1,1} - 1/w_1)_+ , \ldots, (\sigma_{1, r_1} -
  1/w_{r_1})_+]$ and $\bar \Sigma_2 = \diag[ (\sigma_{2,1} - 1/w_1)_+
  , \ldots, (\sigma_{2, r_1} - 1/w_{r_2})_+]$. We want to prove that
  $\norm{A - B}_2^2 - \norm{\bar A - \bar B}_2^2 \geq 0$. First use
  the decomposition
  \begin{align*}
    \norm{A - B}_2^2 - &\norm{\bar A - \bar B}_2^2 = \norm{A}_2^2 -
    \norm{\bar A}_2^2 + \norm{B}_2^2 - \norm{\bar
      B}_2^2 - 2 \inr{A, B} + 2 \inr{\bar A, \bar B} \\
    &= \sum_{j=1}^{r_1} \sigma_{1, j}^2 - \sum_{j=1}^{\bar r_1} \Big(
    \sigma_{1, j} - \frac{1}{w_j} \Big)^2 + \sum_{j=1}^{r_2}
    \sigma_{2, j}^2 - \sum_{j=1}^{\bar r_2} \Big( \sigma_{2, j}^2 -
    \frac{1}{w_j} \Big)^2 \\
    & \quad - 2 \big(\inr{A, B} - \inr{\bar A, \bar B}\big),
  \end{align*}
  where we take $\bar r_1$ such that $\sigma_{1,j} > 1 / w_j$ for $j
  \leq \bar r_1$ and $\sigma_{1,j} \leq 1 / w_j$ for $j \geq \bar r_1
  + 1$, and similarly for $\bar r_2$. We decompose
  \begin{equation}
    \label{eq:dec-inr-lemma-lip}
    \inr{A, B} - \inr{\bar A, \bar B} = \inr{A - \bar A, B - \bar B}
    + \inr{\bar A, B - \bar B} + \inr{A - \bar A, \bar B} 
  \end{equation}
  Using von Neumann's trace inequality $\inr{X,Y} \leq
  \sum_{j} \sigma_j(X) \sigma_j(Y)$ (see for instance \cite{MR832183},
  Section~7.4.13), it follows for the first term
  of~\eqref{eq:dec-inr-lemma-lip} that
  \begin{equation*}
    \inr{A - \bar A, B - \bar B}
    \leq \sum_{j=1}^{r_1 \wedge r_2} (\Sigma_{1} - \bar
    \Sigma_{1})_{j,j} (\Sigma_{2} -  \bar \Sigma_{2})_{j,j}.
  \end{equation*}
  Using the same argument for the two other terms
  of~\eqref{eq:dec-inr-lemma-lip}, we obtain
  \begin{align*}
    \inr{A, B} - \inr{\bar A, \bar B} &\leq \sum_{j=1}^{r_1 \wedge
      r_2} \Big( (\Sigma_{1} - \bar \Sigma_{1})_{j,j}(\Sigma_{2} -
    \bar \Sigma_{2})_{j,j} + (\bar \Sigma_{1})_{j,j}(\Sigma_2
    - \bar  \Sigma_{2})_{j,j} \\
    & +(\Sigma_{1} - \bar \Sigma_{1})_{j,j} (\bar
    \Sigma_{2})_{j,j}\Big), \\
  \end{align*}
 We explore the case $r_1 \leq r_2$ and $\bar
  r_1 \leq \bar r_2$; the other cases follow the same argument. We have 
  \begin{equation*}
    \inr{A, B} - \inr{\bar A, \bar B} \leq \sum_{j = 1}^{\bar r_1}
    \frac{\sigma_{1, j}}{w_j} +   \Big(\sigma_{2, j} - \frac{1}{w_j}
    \Big) \frac{1}{w_j} + \sum_{j = \bar r_1 +
      1}^{r_1} \sigma_{1, j} \sigma_{2, j},
  \end{equation*}
  so, an easy computation leads to
  \begin{align*}
    \norm{A - B}_2^2 - \norm{\bar A - \bar B}_2^2 &\geq \sum_{j=\bar
      r_2 + 1}^{r_1} \sigma_{1, j}^2 + \sum_{j=\bar r_2 + 1}^{r_2}
    \sigma_{2, j}^2 - 2 \sum_{j=\bar r_2 + 1}^{r_1} \sigma_{1, j}
    \sigma_{2, j} \\
    & \quad + \sum_{j=\bar r_1 + 1}^{\bar r_2} \Big( \sigma_{1, j}^2 -
    2 \sigma_{1, j} \sigma_{2, j} + \frac{2 \sigma_{2, j}}{w_j} -
    \frac{1}{w_j^2} \Big).
  \end{align*}
  We obviously have $\sum_{j=\bar r_2 + 1}^{r_1} \sigma_{1, j}^2 +
  \sum_{j=\bar r_2 + 1}^{r_2} \sigma_{2, j}^2 - 2 \sum_{j=\bar r_2 +
    1}^{r_1} \sigma_{1, j} \sigma_{2, j} \geq 0$. By definition of
  $\bar r_2$ and $\bar r_1$, we have $\sigma_{1, j} \leq 1 / w_j <
  \sigma_{2, j}$ for any $j = \bar r_1+1, \ldots, \bar r_2$. Hence, we
  have
  \begin{equation*}
    \sigma_{1, j}^2 - 2 \sigma_{1, j} \sigma_{2, j} + \frac{2
      \sigma_{2, j}}{w_j} - \frac{1}{w_j^2} = (\sigma_{1, j} -
    2\sigma_{2, j} + 1/w_j) (\sigma_{1, j} - 1/w_j) \geq 0,
  \end{equation*}
  which concludes the proof of Proposition~\ref{prop:S-w-lipshitz}.
\end{proof}

\begin{proof}[Proof of Theorem~\ref{prop:existence-and-unicity}]
  Consider the sequence $(A^k)_{k \geq 0}$ defined
  in~\eqref{eq:iterations}. Using Proposition~\ref{prop:S-w-lipshitz}
  we have for any $k\geq1$
  \begin{align*}
    \norm{A^{k+1} - A^{k}}_2 &= \frac{1}{(1 + \tau)}
    \norm{S_{\lambda}^w(\po(A_0) + \po^\perp(A^k)) -
      S_{\lambda}^w(\po(A_0) + \po^\perp(A^{k-1}))}_2 \\
    &\leq \frac{1}{(1 + \tau)} \norm{\po^\perp(A^k) -
      \po^\perp(A^{k-1})}_2 \leq \frac{1}{(1 + \tau)} \norm{A^k - A^{k-1}}_2,
  \end{align*}
  so that $\norm{A^{k+1} - A^{k}}_2 \leq (1 +
      \tau)^{-k} \norm{A^{1} - A^{0}}_2.$
  This proves that $\sum_{k \geq 0} \norm{A^{k+1} - A^{k}}_2 <
  +\infty$, so the limit of $(A^k)_{k \geq 0}$ exists and is given by
  \begin{equation*}
    A^\infty = \sum_{k \geq 0} (A^{k+1} - A^{k})+A^0.
  \end{equation*}
  Now, by continuity of $S_\lambda^w$ and $\cP_{\Omega}^\perp$, taking
  the limit on both sides of~\eqref{eq:iterations}, we obtain that
  $A^\infty$ satisfies the fixed-point equation
  \begin{equation*}
    A^\infty = \frac{1}{1 + \tau}
    S_{\lambda}^w(\cP_{\Omega}^\perp(A^\infty) + \cP_\Omega(A_0)),
  \end{equation*}
  so we have found at least one solution. Let us show now that it is
  unique, so that $\hat A_\lambda^w = A^\infty$: consider a
  matrix $B$ satisfying the same fixed point equation. We have 
  \begin{align*}
    \norm{B - A^\infty}_2 &= \frac{1}{(1 + \tau)^2}
    \norm{S_{\lambda}^w(\po(A_0) + \po^\perp(B)) -
      S_{\lambda}^w(\po(A_0) + \po^\perp(A^\infty))}_2 \\
    &\leq \frac{1}{(1 + \tau)} \norm{\po^\perp(B) -
      \po^\perp(A^\infty)}_2 \leq \frac{1}{(1 + \tau)} \norm{B - A^\infty}_2,
  \end{align*}
  therefore $B=A^\infty$.   
\end{proof}

\begin{proof}[Proof of Theorem~\ref{thm:algorithm_convergence}]
  We know from the proof of Theorem~\ref{prop:existence-and-unicity}
  that
  \begin{equation*}
    \| \hat A_\lambda^w - A^n \|_2 = \norm{\sum_{k \geq n} (A^{k+1} -
      A^k)}_2 \leq \sum_{k \geq n} \frac{1}{(1 +
      \tau)^{k}} \norm{A^{1} - A^{0}}_2,
  \end{equation*}
  leading to the conclusion.
\end{proof}

\bibliographystyle{plain}


\end{document}